\documentclass[lettersize,journal]{IEEEtran}
\usepackage{amsmath,amsfonts}
\usepackage{algorithmic}
\usepackage{algorithm}
\usepackage{array}
\usepackage[caption=false,font=normalsize,labelfont=sf,textfont=sf]{subfig}
\usepackage{textcomp}
\usepackage{stfloats}
\usepackage{amsthm}
\usepackage{url}
 \usepackage{booktabs}
\usepackage{verbatim}
\usepackage{graphicx}
\usepackage{cite}
\usepackage{makecell}
\newtheorem{proposition}{Proposition}
\usepackage{amssymb}
\usepackage{hyperref}
\begin{document}

\title{Lightweight Quantum Agent  for Edge Systems: Joint PQC and NOMA Resource Allocation}

\author{Yongtao~Yao, Wenjing Xiao, 
        Miaojiang~Chen, Anfeng Liu, Zhiquan Liu, Min Chen~\IEEEmembership{Fellow,~IEEE}, Ahmed Farouk, H. Herbert Song~\IEEEmembership{Fellow,~IEEE}

\thanks{Manuscript received xx, 2026; revised xx, 2026. This work was supported in part by the National Natural Science Foundation of China (62462002) and partially supported by the Natural Science Foundation of Guangxi, China (Nos. 2025GXNSFAA069958, 2025GXNSFBA069394), and the Key Research and Development Program of Guangxi (No. AD25069071). (Corresponding author:  Miaojiang Chen, H. Herbert Song)
}
\thanks{Yongtao Yao, Wenjing Xiao,  Miaojiang Chen are with the School of Computer, Electronics and Information, Guangxi University, and also with the Guangxi Key Laboratory of Multimedia Communications and Network Technology, Nanning 530004, China. E-mail: yongtao@st.gxu.edu.cn, wenjingx@gxu.edu.cn, mjchen\_cs@gxu.edu.cn.}

\thanks{Anfeng Liu is with School of Computer Science and Engineering, Central South University,
	Changsha 410083, China. E-mail:afengliu@mail.csu.edu.cn.}

\thanks{Zhiquan Liu is with College of Cyber Security, Jinan University, Guangzhou 510632, China. E-mail:zqliu@jnu.edu.cn.}

\thanks{Min Chen  is with the School of Computer Science and Engineering, South China University of Technology, Guangzhou 510006, China, and also with the Pazhou Laboratory, Guangzhou 510330, China. E-mail: minchen@ieee.org.}

\thanks{Ahmed Farouk is with the Faculty of Computers and Artificial Intelligence, Hurghada University, Hurghada 83523, Egypt (e-mail: ahmed.farouk@sci.svu.edu.eg).}

\thanks{H. Herbert Song is  with the Department of Information Systems, University of Maryland, Baltimore County (UMBC), Baltimore, MD 21250 USA. E-mail: h.song@ieee.org.}
}


\markboth{Journal of \LaTeX\ Class Files,~Vol.~14, No.~8, February~2026}%
{Yao \MakeLowercase{\textit{\textit{et al.}}}: Low Complexity Online Optimization for Quantum Secure NOMA MEC}



\maketitle

\begin{abstract}
In the context of quantum secure scenarios, existing research on mobile edge devices and intelligent computing 
and edge (ICE) systems based on the Non-Orthogonal Multiple Access (NOMA) communication model have overlooked the energy consumption overhead of Post-Quantum Cryptography (PQC) modules, and the high complexity of traditional resource allocation algorithms fails to meet the demands of real-time decision-making. To address these challenges, this paper proposes a lightweight agentic AI framework designed for online joint optimization within ICE-enabled mobile devices. The scheme constructs a multi-stage stochastic Mixed Integer Nonlinear Programming (MINLP) model that incorporates static power-consumption constraints for PQC modules. Based on Lyapunov optimization theory, the long-term optimization problem is decoupled, and a linear complexity algorithm is proposed to solve the nonconvex challenges of NOMA power allocation
. Simulation results verify that the proposed scheme significantly improves computational throughput while ensuring system queue stability and energy consumption constraints. Compared with traditional Successive Convex Approximation (SCA) algorithms, the complexity is reduced to $\mathcal{O}(N)$, achieving a speedup of approximately 46 times when the number of devices $N=35$, thereby meeting the real-time decision-making requirements in dynamic wireless environments.
\end{abstract}


\section{INTRODUCTION}\label{tab:INTRODUCTION}
\IEEEPARstart{W}{ith}  the intelligent upgrading of consumer electronic terminals and the ubiquitous interconnection \cite{10819462}, computationally intensive applications such as augmented reality, real-time video inference, and large language model inference have become dominant \cite{10878496}. However, limited by local processing capabilities and battery capacity, mobile devices struggle to undertake these tasks independently under strict energy-consumption constraints \cite{8016573}. ICE, which deploys servers with strong computational capabilities at the wireless access side, allows users to offload tasks to edge nodes for execution. This significantly reduces communication latency and alleviates terminal energy consumption, establishing ICE as a key enabling technology for 5G and future networks \cite{7488250, 10959110, 7879258, 10373014}.

In the study of offloading strategies and resource management for multi-user ICE systems, traditional approaches often rely on Time Division Multiple Access (TDMA) architectures \cite{8334188}. These methods avoid interference through orthogonal resource allocation and jointly optimize offloading decisions and computational resources \cite{7762913}. However, the orthogonal nature of TDMA limits spectral efficiency: as device density increases, the available resources per user decrease, creating a throughput bottleneck, which is particularly pronounced in Internet of Things (IoT) scenarios. NOMA technology, through power-domain superposition coding and Successive Interference Cancellation (SIC), allows multiple users to share the same time-frequency resources, theoretically breaking the capacity limits of orthogonal access \cite{8114722, 7842433}. Introducing NOMA into the ICE offloading architecture can support concurrent multi-user transmissions without increasing bandwidth, thereby enhancing the system's computational throughput \cite{10884688}. Research indicates that under the same power conditions, the system throughput of NOMA-MEC can be 20\% to 40\% higher than that of TDMA-MEC, with the gains becoming more significant as user density increases \cite{8537962}.

Despite this potential, joint multi-user MEC optimization based on NOMA still faces three theoretical challenges that existing research has not fully addressed. First, binary offloading yields an MINLP with a $2^N$ solution space , and relaxation often fails to guarantee integrality. Secondly, channel fading and task randomness introduce uncertainty regarding future information. Conventional multi-stage programming requires knowledge of statistical distributions, which are difficult to estimate accurately in practice. Thirdly, in NOMA mode, the transmit powers of users are strongly coupled, making the power allocation subproblem highly nonconvex. SCA-based methods converge to KKT points but incur $\mathcal{O}(N^{3.5})$ per-iteration complexity, unsuitable for real-time decisions .

To address the challenges of online decision making and real-time requirements, Lyapunov stochastic network optimization theory provides a systematic set of tools \cite{neely2010stochastic}: by constructing a Lyapunov function and minimizing its single slot Drift Plus Penalty (DPP) upper bound, the multi-stage stochastic problem can be decoupled into per-frame deterministic subproblems. This ensures queue stability and satisfaction of power constraints without requiring prior statistical information. The Lyapunov-guided Deep Reinforcement Learning (LyDROO) framework proposed in \cite{9449944} deeply integrates DPP with deep reinforcement learning, generating relaxed offloading strategies via a DNN and combining them with a model-driven resource allocation module for precise evaluation, thereby significantly improving convergence speed and decision efficiency. DRL methods have demonstrated good adaptability to dynamic environments in wireless resource scheduling \cite{8714026, 8103164}, providing an essential foundation for constructing online optimization frameworks.







Meanwhile, the development of quantum computing poses a systemic threat to existing wireless communication security systems. Shor's algorithm demonstrates that universal quantum computers can solve RSA and elliptic curve cryptography, among other public key systems, in polynomial time \cite{shor1999polynomial, preskill2018quantum}. In ICE scenarios where devices upload sensitive tasks over open wireless channels, user privacy and data integrity are severely threatened if the uplink is subjected to quantum attacks. PQC hardware security modules, included in NIST's standardization framework, are currently recognized as a viable path against quantum attacks \cite{ducas2018crystals}.

Based on the background outlined above, this paper proposes a low-complexity online joint optimization scheme for quantum-secure NOMA ICE systems. The main contributions are as follows:

\begin{enumerate}
    \item \textbf{Quantum Aware System Modeling}: Based on the Non Orthogonal Multiple Access communication model, the constant power consumption $P_{\mathrm{qlc}}$ of the PQC security module is explicitly incorporated into the optimization framework of the multi-user ICE system. A multi-stage stochastic MINLP model is established that incorporates binary offloading decisions, local CPU frequency control, and NOMA transmit power allocation. By introducing a virtual energy queue, the long-term average energy consumption constraint is transformed into a queue-stability problem amenable to online tracking.

    \item \textbf{Linear Complexity Power Allocation Algorithm}: A NOMA power allocation algorithm based on greedy backward induction is proposed. By deriving a recursive analytical expression for the inter-user interference terms, the global nonconvex power allocation problem is strictly decomposed into $|\mathcal{M}_1|$ independent single-variable convex subproblems, each possessing a closed-form optimal solution. The overall computational complexity of the algorithm is $\mathcal{O}(N)$, achieving approximately a 46× speedup compared to the baseline SCA algorithm with complexity $\mathcal{O}(K_{\mathrm{iter}} \cdot N^{3.5})$ for $N=35$, effectively meeting the latency requirements for real time per frame decision making.

    \item \textbf{Quantitative Analysis of Quantum Power Impact}: Through systematic simulations, the influence of $P_{\mathrm{qlc}}$ on the optimal offloading ratio, virtual energy queue evolution, and system throughput is quantitatively revealed. The adaptive adjustment mechanism of the Lyapunov controller in response to perturbations from quantum security power consumption is elucidated, providing a basis for the practical deployment and parameter configuration of quantum secure ICE systems.
\end{enumerate}

The remainder of this paper is organized as follows: Section \ref{sec:RELATED WORK} reviews related work; Section \ref{sec:SYSTEM MODEL} describes the system model and problem formulation; Section \ref{sec:LyDROO Framework} elaborates on the LyDROO-based online lightweight agentic AI framework; Section \ref{sec:Power Allocation} proposes the low complexity power allocation algorithm; Section \ref{sec:evaluation} presents simulation results and analysis; Section \ref{sec:CONCLUSION} concludes the paper.

\section{RELATED WORK}\label{sec:RELATED WORK}

\subsection{Resource Management in NOMA-enabled ICE Systems.}
By deploying abundant computational resources at the wireless access edge, ICE enables mobile devices to offload computationally intensive, latency-sensitive tasks to edge nodes, making it a key enabling technology for mitigating terminal energy bottlenecks and empowering 5G and future networks. Hasan \textit{\textit{et al.}} \cite{10415079} summarized relevant architectures and challenges. In early research on multi-user ICE resource allocation, traditional architectures mostly relied on TDMA or FDMA to jointly optimize offloading decisions and computational resources. However, the spectral efficiency of orthogonal access is limited, easily forming throughput bottlenecks in scenarios with massive concurrent access from IoT devices. To overcome the orthogonal capacity limits from an information theoretic perspective, NOMA has been introduced into ICE architectures. Shi \textit{et al.} \cite{9893789} pioneered the study of latency minimization problems in NOMA-ICE systems. Wang \textit{et al.} \cite{10101816} used DDPG to minimize latency in NOMA ICE via iterative optimization.  Li \textit{et al.} \cite{10413502} analyzed the computational complexity of such joint optimization problems. Furthermore, the strong coupling in NOMA renders the power allocation subproblem highly nonconvex. While existing algorithms based on SCA can converge to KKT points. Convex optimization algorithms are endowed with well-established foundational convergence properties, yet their per-iteration superlinear computational complexity of $\mathcal{O}(N^{3.5})$ falls short of meeting millisecond-level real-time decision-making demands in fast-fading wireless channel scenarios. Shi \textit{et al.} \cite{11433062} proposed a convex evolutionary alternating optimization for joint user selection and resource allocation in multi-cell NOMA-MEC, achieving near-optimal service satisfaction within few iterations. Zhang \textit{et al.} \cite{9194041} addressed the joint location and power optimization for NOMA-UAV networks with mobility-dependent decoding order, developing an SCA-based iterative scheme that maximizes the sum rate with asymptotic performance guarantees.

\subsection{Online Optimization and Deep Reinforcement Learning in Mobile Edge Computing.}
Due to the randomness of wireless channel fading and task arrivals \cite{11458888,xu2024harvesting}, ICE systems face significant uncertainty regarding future state information during actual operation. Traditional Markov decision processes or dynamic programming methods are often limited by the lack of environmental statistical distribution and the curse of dimensionality. To address this issue, Zheng \textit{et al.} \cite{10019272} studied the total computation delay minimization problem in wireless-powered multi-access edge computing networks, jointly optimizing task offloading decisions and wireless resource allocation, and proposed a low-complexity online offloading algorithm based on deep reinforcement learning. Wang \textit{et al.} \cite{9999714} addressed the delay optimization problem for UAV-assisted terahertz edge computing networks, proposing deep reinforcement learning algorithms based on DDQN. Nie \textit{et al.}\cite{9563249} focused on the power optimization problem for multi-UAV-assisted edge computing systems and proposed a semi-distributed multi-agent federated reinforcement learning algorithm with privacy protection mechanisms. To address the inefficiency of traditional heuristic algorithms for solving nonconvex MINLP problems, Bie \textit{et al.} \cite{9449944} proposed the LyDROO framework, which deeply integrates Lyapunov optimization with DRL to generate relaxed offloading policies via DNNs, thereby significantly improving decision-making and convergence efficiency. Chen \textit{et al.} \cite{11450522} further extended this line by proposing a DDPG-Attention-based algorithm that integrates Lyapunov optimization and DRL for joint resource allocation and trajectory optimization in NOMA-enabled HAP-UAV-MEC systems. Liu \textit{et al.} \cite{10736570} further addressed DNN partitioning and task offloading in dynamic vehicular networks by proposing a Lyapunov-guided multi-agent diffusion-based DRL algorithm, integrating convex optimization subroutines for enhanced learning efficiency.By combining the classic Actor-Critic architecture.  With an online computing framework that integrates model-driven evaluation with data-driven generation, this paradigm has become an essential approach for solving dynamic MEC joint scheduling problems.

\subsection{Quantum Security Challenges in Edge Networks.}
Meanwhile, the rapid advancement of quantum computing poses a systemic threat to existing security foundations. Shor \textit{et al.} \cite{shor1999polynomial} demonstrated that a universal quantum computer could break public-key cryptosystems in polynomial time. Preskill \textit{et al.} \cite{preskill2018quantum} confirmed this long-term threat despite immature hardware. In I scenarios, devices transmit sensitive data over open wireless channels, making them vulnerable to eavesdropping and quantum attacks. Prateek \textit{et al.} \cite{prateek2023quantum} examined quantum vulnerabilities in edge communication systems. Recent works have explored quantum-safe edge communication schemes and hybrid quantum-classical reinforcement learning for secure resource allocation, yet the integration of PQC modules into resource-constrained edge devices remains an open challenge.

PQC algorithm hardware security modules have been included in NIST standardization. Chen \textit{et al.} \cite{chen2016report} detailed the NIST post-quantum cryptography standardization process. As a recognized path against quantum attacks. However, for resource-constrained terminals, PQC modules add fixed circuit power overhead. Most existing online optimization studies for NOMA-based I overlook this critical secure energy cost.

Therefore, current research exhibits significant limitations on two levels: Model based optimization methods, while providing theoretical guarantees, have an iterative complexity as high as $\mathcal{O}(N^{3.5})$, making them difficult to adapt to real-time decision making needs in dynamic wireless environments; Conversely, purely data driven DRL methods, although possessing dynamic adaptability, suffer from low sample efficiency and slow convergence due to neglecting the system's physical model. Furthermore, existing work generally fails to incorporate the constant power consumption $P_{\mathrm{qlc}}$ of the quantum security module into system modeling, thereby undermining the effectiveness of the optimization framework in quantum secure scenarios.

\section{SYSTEM MODEL}\label{sec:SYSTEM MODEL}

We consider an edge server (ES) assisting $N$ wireless devices (WDs) with computation within equal-length time frames $T$. In the $t$-th time frame, let $A_i^t$ denote the amount of raw task data arriving at the data queue of the $i$-th WD. We assume that the arrivals $A_i^t$ follow a general independent and identically distributed (i.i.d.) process with bounded second moments.

\begin{figure}[htbp]
    \centering
    \includegraphics[width=0.45\textwidth]{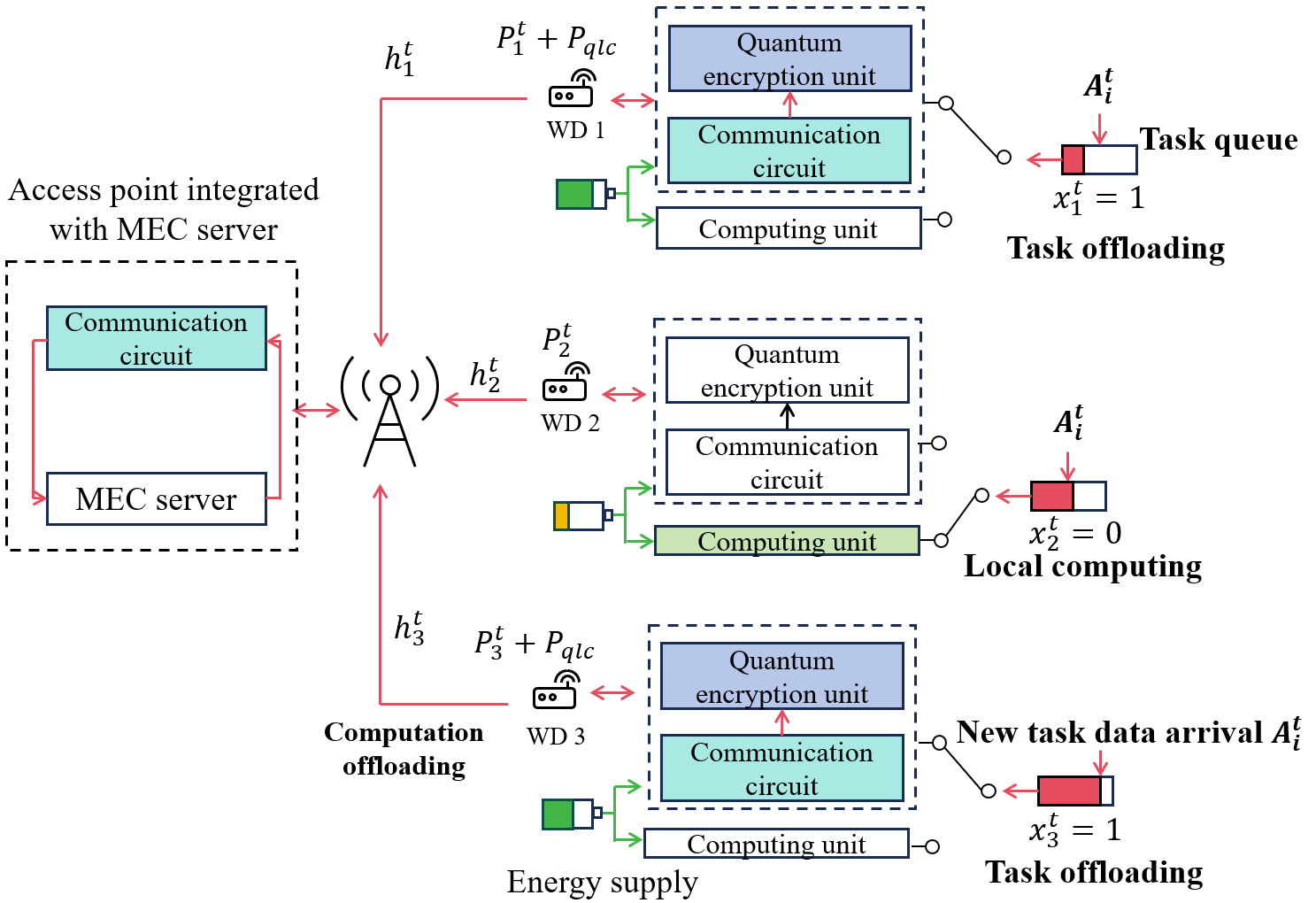} 
    \caption{Quantum secure NOMA access mobile edge computing network.}
    \label{fig: bg}
\end{figure}

Let $h_i^t$ represent the channel gain between the $i$-th WD and the ES. Under the block fading assumption, $h_i^t$ remains constant within one time frame but varies independently across different time frames. In the $t$-th time frame, it is assumed that the target WD $i$ processes $D_i^t$ bits of data and generates computation output at the end of the time frame. We believe all WDs adopt a binary computation offloading rule, meaning that the raw data is either processed entirely locally or entirely offloaded to the ES for remote processing. We use a binary variable $x_i^t$ to denote the offloading decision, where $x_i^t = 1$ and $0$ indicate that WD $i$ performs computation offloading and local computation, respectively. For example, in Figure \ref{fig: bg}, wireless devices 1 and 3 perform task offloading, while wireless device 2 performs local computation.

For Local Computation ($x_i^t = 0$)
When a WD processes data locally, we denote the local CPU frequency as $f_i^t$, which is upper-bounded by $f_i^{\max}$. The amount of raw data processed locally and the energy consumed within a time frame are calculated, respectively, as:
\begin{equation}
D_{i,L}^{t}=f_{i}^{t}T/\phi , E_{i,L}^{t}=\kappa \left( f_{i}^{t}\right) ^{3}T, \forall x_{i}^{t}=0,
\end{equation}
where the parameter $\phi > 0$ represents the number of computation cycles required to process one bit of raw data, and $\kappa > 0$ denotes the computational energy efficiency parameter.

For Computation Offloading ($x_i^t = 1$)
When wireless devices offload tasks to the edge server for execution, they simultaneously transmit data to the edge server using NOMA technology over the shared bandwidth $W$. Let $P_i^t \in [0, P_i^{\max}]$ denote the transmit power. Assuming the channel gains are sorted in descending order ($h_1^t \geq h_2^t \geq \cdots \geq h_N^t$), the edge server employs the SIC technique to decode the signals. In this case, the Signal to Interference plus Noise Ratio (SINR) for the $i$-th wireless device is:

\begin{equation}\text{SINR}_i^t = \frac{P_i^t h_i^t}{N_0 + \sum_{j=i+1}^{N} P_j^t h_j^t}, \quad \forall x_i^t = 1,\end{equation}
where $N_0$ is the noise power. The parameter $\eta_q \in (0, 1]$ is defined as the effective payload factor after quantum-secure coding, which characterizes the proportion of the original task data in the total transmitted bitstream after the complex encapsulation and signature redundancy introduced by PQC algorithms, reflecting the impact of security overhead on communication effectiveness.

\begin{equation}
D_{i,O}^t = \frac{\eta_q W}{v_u} \log_2 \left(1 +\text{SINR}_i^t \right), \quad \forall x_i^t = 1,
\end{equation}
where $v_u \geq 1$ is the communication overhead coefficient. Considering that the quantum security module introduces an additional constant circuit power consumption $P_{i,qlc} > 0$, the total energy consumed during data offloading is: 
\begin{equation}
E_{i, O}^t =( P_i^t + P_{i,qlc})T, \quad \forall x_i^t = 1,
\end{equation}
Combining the above two execution modes, within the $t$-th time slot, the total computation rate $r_i^t$ and total power consumption $e_i^t$ of the $i$-th wireless device can be expressed respectively as:

\begin{align}
   r_{i}^{t} 
   &=\frac{D_{i}^{t}}{T} \\
   &=\frac{\left(1-x_{i}^{t}\right)D_{i,L}^{t}+x_{i}^{t}D_{i,O}^{t}}{T}\nonumber\\
   &=\frac{\left(1-x_{i}^{t}\right) f_{i}^{t}}{\phi}\nonumber+x_{i}^{t}\frac{\eta_qW}{v_{u}}\log_2\left(1 +\text{SINR}_i^t \right) ,\nonumber
\end{align}

\begin{align}
      e_i^t 
      &= \frac{E_i^t}{T} \\
      &= \frac{(1-x_i^t)E_{i,L}^{t}+x_i^t E_{i,O}^{t}}{T}\nonumber\\
      &= (1-x_i^t) \kappa (f_i^t)^3 + x_i^t (P_i^t+ P_{i,qlc}).\nonumber
\end{align}



Let $Q_i(t)$ denote the length of the pending data queue for the $i$-th wireless device at the beginning of the $t$-th time slot. The queue dynamics can be modeled as:
$$Q_i(t+1) = \max\{Q_i(t) - \tilde{D}_i^t + A_i^t, 0\}, \quad i = 1, 2, \cdots,$$
where $\tilde{D}_i^t = \min(Q_i(t), D_i^t)$, and the initial queue state is $Q_i(1) = 0$. By introducing the data causality constraint $D_i^t \leq Q_i(t)$, it is guaranteed that $Q_i(t) \geq 0$ holds for any time instant $t$, thus the queue dynamics model can be simplified to:
$$Q_i(t+1) = Q_i(t) - D_i^t + A_i^t, \quad i = 1, 2, \cdots.$$
A discrete time queue $Q_i(t)$ is said to be strongly stable if its time average queue length satisfies $\lim_{K \to \infty} \frac{1}{K} \sum_{t=1}^{K} \mathbb{E}[Q_i(t)] < \infty$. \cite{georgiadis2006resource} We introduce $N$ virtual energy queues $\{Y_{i}(t)\}_{i=1}^{N}$ for each wireless device (WD), with one virtual queue corresponding to each device. Specifically, we set the initial value $Y_{i}(1)=0$ and update the queue as follows:
$$Y_{i}(t+1)=\operatorname* {max}\left( Y_{i}(t)+\nu e_{i}^{t}-\nu \gamma _{i},0\right),\ i=1,2,\cdots $$
where $\nu$ is the scaling factor, and $\gamma_{i}$ is the device's average energy consumption threshold.

In this paper, we aim to design an online computation offloading algorithm that maximizes the long-term average weighted sum computation rate across all WDs, subject to constraints on long-term data queue stability and average power. To achieve this goal, we formulate this problem as the following multi-stage stochastic MINLP problem:

\begin{subequations}
\label{prob:optimization_problem} 
\begin{align}
&\underset{\mathbf{x},\mathbf{f},\mathbf{p}}{\operatorname{maximize}}
\quad \lim_{K\rightarrow\infty}\frac{1}{K} \sum_{t=1}^{K}\sum_{i=1}^{N}c_{i}r_{i}^{t} \label{obj:total_reward} \\ 
& \text{subject to} \nonumber \\
& \qquad (1-x_{i}^{t})f_{i}^{t}/\phi + x_i^t\frac{\eta_qW}{v_u} \log_2\Bigl( \notag \\
& \qquad \qquad 1+\text{SINR}_i^t\Bigr) \leq Q_{i}(t), \quad \forall i,t, \label{con:data_causality} \\
& \qquad \lim_{K\rightarrow\infty}\frac{1}{K} \sum_{t=1}^{K}
\mathbb{E}\Bigl[ (1-x_{i}^{t})\kappa(f_{i}^{t})^{3} \notag \\
& \qquad \qquad + x_{i}^{t}(P_i^t+ P_{i,qlc})\Bigr] \leq \gamma_{i}, \quad \forall i,t, \label{con:energy} \\
& \qquad \lim_{K\rightarrow\infty}\frac{1}{K} \sum_{t=1}^{K}
\mathbb{E}\left[ Q_{i}(t)\right] < \infty, \quad \forall i, \label{con:stability} \\
& \qquad x_{i}^{t} \in \{0,1\}, \quad \forall i,t, \label{con:variable_domain}\\
& \qquad 0 \leq f_{i}^{t} \leq f_{i}^{\max}, \quad \forall i,t, \label{con:instant_power} \\
& \qquad 0 \leq P_i^t \leq P_i^{\text{max}}, \quad \forall i,t. \label{con:power_max}
\end{align}
\end{subequations}

Where $c_{i}$ represents the fixed weight for the $i$-th wireless device. In the above formulation, the data causality constraint ensures that the amount of data processed in any given time frame does not exceed the amount of data currently in the queue. It should be noted that at the optimal solution, if $x_{i}^{t}=0$ (i.e., local computation is performed), then the transmit power $P_i^t=0$ and the edge computing rate term becomes zero; similarly, if $x_{i}^{t}=1$ (i.e., computation offloading is performed), the local computing frequency $f_{i}^{t}=0$ must hold. The average power constraint limits the long-term energy consumption of the devices, where $\gamma_{i}$ is the power threshold. The data queue stability constraint guarantees the strong stability of the queues. We normalize the slot duration to 1, i.e., \(T = 1\). Consequently, the amount of data processed or transmitted within one slot is numerically equal to the corresponding rate. In the following, we will adopt the LyDROO framework proposed by Bi \textit{et al.} \cite{9449944} to solve the above problem.

\section{Lightweight Agentic AI Framework}\label{sec:LyDROO Framework}
\subsection{Lyapunov Framework}\label{sec:lyapunov}
To jointly control the data queues and virtual energy queues, we define the overall system queue backlog state as $Z(t)=\{Q(t), Y(t)\}$. Next, we introduce the Lyapunov function $L(Z(t))$ and its one step drift $\Delta L(Z(t))$ as follows:
\begin{equation}
L\left(\mathbf{Z}(t)\right)=\frac{1}{2}\left(\sum_{i=1}^{N}Q_{i}(t)^{2}+\sum_{i=1}^{N}Y_{i}(t)^{2}\right),
\end{equation}

\begin{equation}
\Delta L\left(Z(t)\right)=\mathbb{E}\left\{L\left(Z(t+1)\right)-L\left(Z(t)\right)\middle| Z(t)\right\}.
\end{equation}

To maximize the time-averaged computation rate while ensuring the stability of the queue $Z(t)$, we adopt the DPP minimization method. Specifically, at each time frame $t$, we seek to minimize an upper bound of the following DPP expression:

\begin{equation}
\Lambda\left(\mathbf{Z}(t)\right)=\Delta L\left(\mathbf{Z}(t)\right)-V\cdot\sum_{i=1}^{N}\mathbb{E}\left\{c_{i}r_{i}^{t}\middle|\mathbf{Z}(t)\right\},
\end{equation}
where $V>0$ is an importance weight used to trade off between queue delay and computation rate (the penalty term). Next, we derive an upper bound for $\Lambda(Z(t))$. First, by squaring both sides of the queue dynamics equations $Q_{i}(t+1)=Q_{i}(t)-D_{i}^{t}+A_{i}^{t}$ and $Y_{i}(t+1)=\max\left(Y_{i}(t)+\nu e_{i}^{t}-\nu\gamma_{i},0\right)$ respectively, we can obtain the evolution relations for $Q_{i}(t)$ and $Y_{i}(t)$:
$$Q_i(t+1)^2 = Q_i(t)^2 + 2Q_i(t)\left(A_i^t-D_i^t\right) + \left(A_i^t-D_i^t\right)^2\notag,$$
$$Y_i(t+1)^2 = Y_i(t)^2 + 2\nu Y_i(t)\left(e_i^t-\gamma_i\right) + \nu ^2\left(e_i^t-\gamma_i\right)^2\notag.$$

Transforming these and summing over all devices yields the difference equation for the queues:

\begin{equation}
\begin{split}
\frac{1}{2}\sum_{i=1}^N Q_i(t+1)^2 - \frac{1}{2}\sum_{i=1}^N Q_i(t)^2 = \\
\sum_{i=1}^N Q_i(t)\left(A_i^t-D_i^t\right) + \frac{1}{2}\sum_{i=1}^N\left(A_i^t-D_i^t\right)^2,
\end{split}
\end{equation}

\begin{equation}
\begin{split}
\frac{1}{2}\sum_{i=1}^N Y_i(t+1)^2 - \frac{1}{2}\sum_{i=1}^N Y_i(t)^2 = \\ 
\nu \sum_{i=1}^N Y_i(t)\left(e_i^t-\gamma_i\right) + \frac{1}{2}\nu^2\sum_{i=1}^N\left(e_i^t-\gamma_i\right)^2.
\end{split}
\end{equation}

To analyze the stability of the two types of queues separately, we define the Lyapunov functions and drift terms for the data queues and the energy queues, respectively:
\begin{equation}
L\left(\mathbf{Q}(t)\right)= \frac{1}{2}\sum_{i=1}^N Q_i(t)^2,
\end{equation}

\begin{equation}
\Delta L\left(\mathbf{Q}(t)\right)=\mathbb{E}\left\{L\left(\mathbf{Q}(t+1)\right)-L\left(\mathbf{Q}(t)\right)\middle|\mathbf{Z}(t)\right\},
\end{equation}

\begin{equation}
L\left(\mathbf{Y}(t)\right)=\frac{1}{2}\sum_{i=1}^{N} Y_{i}(t)^{2},
\end{equation}

\begin{equation}
\Delta L\left(\mathbf{Y}(t)\right)=\mathbb{E}\left\{L\left(\mathbf{Y}(t+1)\right)-L\left(\mathbf{Y}(t)\right)\middle|\mathbf{Z}(t)\right\}.
\end{equation}
Taking the conditional expectation of the above squared difference equations, we can obtain the upper bounds for the two drift terms respectively:
$$\Delta L\left(\mathbf{Q}(t)\right)\leq B_{1}+\sum_{i=1}^{N} Q_{i}(t)\mathbb{E}\left[\left(A_{i}^{t}-D_{i}^{t}\right)\middle|\mathbf{Z}(t)\right],$$
where the constant $B_{1}$ arises from the upper bound on the expectation of the square term, and $r_{i}^{\max}=\mathbb{E}\left[\frac{\eta_qW}{v_{u}}\log_{2}\left(1+\frac{P_{i}^{\max}h_{i}^{t}}{N_0}\right)\right]$, derived as follows:
\begin{align*}
&\frac{1}{2}\sum_{i=1}^{N}\mathbb{E}\left[\left(A_{i}^{t}-D_{i}^{t}\right)^{2}\right] \leq \frac{1}{2}\sum_{i=1}^{N}\mathbb{E}\left[\left(A_{i}^{t}\right)^{2}+\left(D_{i}^{t}\right)^{2}\right]\\
&\leq \frac{1}{2}\sum_{i=1}^{N}\left(\eta_{i}+\left[T\max\left\{f_{i}^{\max}/\phi,r_{i}^{\max}\right\}\right]^{2}\right)= B_{1}.
\end{align*}

Similarly, for the drift term of the virtual energy queues, we have:
$$\Delta L\left(\mathbf{Y}(t)\right)\leq B_{2}+\nu \sum_{i=1}^{N} Y_{i}(t)\mathbb{E}\left[e_{i}^{t}-\gamma_{i}\middle|\mathbf{Z}(t)\right],$$
where the constant $B_{2}$ is defined as:

\begin{align*}
&\frac{1}{2}\nu^2\sum_{i=1}^{N}\mathbb{E}\left[\left(e_{i}^{t}-\gamma_{i}\right)^{2}\right] \\
&\leq \frac{1}{2}\nu^2\sum_{i=1}^{N}\left[\left(\max\left\{\kappa\left(f_{i}^{\max}\right)^{3},P_{i}^{\max}+ P_{i,qlc}\right\}\right)^{2}+\gamma_{i}^{2}\right]\\
&= B_{2}.
\end{align*}

Let the constant term be $\hat{B}=B_{1}+B_{2}$. Adding the drift upper bounds for the data queues and energy queues yields the upper bound for the one-step drift of the overall system state:

\begin{align*}
\Delta L\left(\mathbf{Z}(t)\right) 
&=\Delta L\left(\mathbf{Y}(t)\right) +\Delta L\left(\mathbf{Q}(t)\right) \\
&\leq\hat{B}+\sum_{i=1}^{N} Q_{i}(t)\mathbb{E}\left[\left(A_{i}^{t}-D_{i}^{t}\right)\middle|\mathbf{Z}(t)\right]\\
&+\nu\sum_{i=1}^{N} Y_{i}(t)\mathbb{E}\left[e_{i}^{t}-\gamma_{i}\middle|\mathbf{Z}(t)\right] .\nonumber
\end{align*}

Finally, substituting the total drift upper bound into the DPP expression $\Lambda\left(\mathbf{Z}(t)\right)$ and combining terms yields:

\begin{align*}
\Lambda\left(\mathbf{Z}(t)\right)
&= \Delta L\left(\mathbf{Z}(t)\right)-V\cdot\sum_{i=1}^{N}\mathbb{E}\left\{c_{i}r_{i}^{t}\middle|\mathbf{Z}(t)\right\}\\
&\leq \hat{B} + \sum_{i=1}^{N} Q_i(t) \mathbb{E}\left[A_i^t - D_i^t \mid \mathbf{Z}(t)\right] \\
&+ \nu \sum_{i=1}^{N} Y_i(t) \mathbb{E}\left[e_i^t - \gamma_i \mid \mathbf{Z}(t)\right]  \\
&-V \sum_{i=1}^{N} \mathbb{E}\left[c_i r_i^t \mid \mathbf{Z}(t)\right] \nonumber\\
&=\hat{B}+ \sum_{i=1}^N Q_i(t)\mathbb{E}\left[A_i^t \mid \mathbf{Z}(t)\right]  \\
&-\sum_{i=1}^N \mathbb{E}\left[(Q_i(t) + V c_i)r_i^t - \nu Y_i(t)e_i^t \mid \mathbf{Z}(t)\right] \\
&- \nu\sum_{i=1}^N Y_i(t)\gamma_i .\nonumber
\end{align*}

At the $t$-th time frame, according to the principle of opportunistically minimizing an expectation \cite{berger2013statistical}, minimizing the above upper bound is equivalent to maximizing the term containing the control variables in the current time slot, i.e.:
$$\sum_{i=1}^{N}\left(Q_{i}(t)+V c_{i}\right)r_{i}^{t}-\nu\sum_{i=1}^{N}Y_{i}(t)e_{i}^{t}.$$

Combining the above derivations, the original complex multi-stage stochastic problem is successfully decoupled. The system only needs to solve the following per-frame deterministic subproblem at each time frame $t$:

\begin{subequations}
\begin{align}
&\underset{\mathbf{x}^{t},\mathbf{f}^{t},\mathbf{P}^{t}}{\text{maximize}} \sum_{i=1}^{N}\left(Q_{i}(t)+V c_{i}\right)r_{i}^{t}-\nu \sum_{i=1}^{N}Y_{i}(t)e_{i}^{t} \nonumber \\
&\text{subject to} \nonumber \\
&(1 - x_i^t) \frac{f_i^t}{\phi} + x_i^t \frac{\eta_qW}{v_u} \log_2\left(1+\text{SINR}_i^t\right) \leq Q_i(t), \quad \forall i, \\
&0 \leq P_i^t \leq P_i^{\text{max}},\quad\forall i, \\
&0 \leq f_i^t \leq f_i^{\text{max}},\quad\forall i,\\
&x_{i}^{t}\in\{0,1\},\quad\forall i. 
\end{align}
\end{subequations}

\begin{figure*}[htbp]
    \centering
    \includegraphics[width=0.95\textwidth]{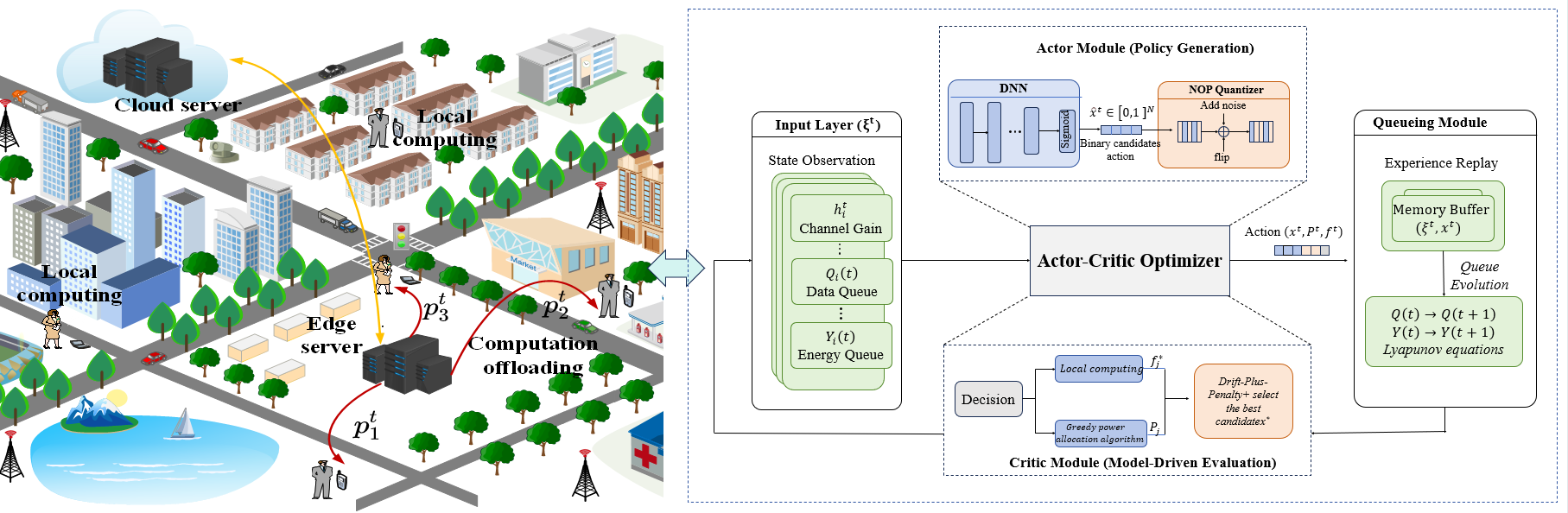} 
    \caption{Overall Architecture of the Proposed LyDROO-Based Online Lightweight Agentic AI  Framework.}
    \label{fig:singlex}
\end{figure*}

\subsection{Actor-Critic Module}\label{sec:AC}
In the $t$-th time frame, to solve the deterministic subproblem, the system first observes the current environmental state $\xi^t = \{h_i^t, Q_i(t), Y_i(t)\}_{i=1}^N$, which consists of the channel gains, data queue backlogs, and virtual energy queue states for all wireless devices (WDs). If the offloading decision variables $\mathbf{x}^t$ are fixed, the problem can be transformed into a pure resource allocation problem. Therefore, the core of solving the original problem lies in finding the optimal binary offloading decision combination.
To avoid the catastrophic computational delay caused by exhaustively enumerating $2^N$ offloading combinations, we adopt a DRL algorithm based on the Actor Critic \cite{6313077} structure, constructing a low complexity policy that maps the input state $\xi^t$ to the optimal offloading action. The algorithm mainly consists of four modules: the actor network module, the critic network module, the policy update module, and the queue update module. The overall structure is shown in Figure \ref{fig:singlex}.

1) Actor Module:  
The Actor module consists of a deep neural network (DNN) and an action quantizer. At the beginning of the \( t \)-th time frame, the DNN receives the environmental observation state \( \xi^t \) as input and outputs a relaxed continuous offloading decision vector \( \hat{x}^t \in [0,1]^N \). The input-output mapping can be expressed as:  
$$\Pi_{\theta^t}: \xi^t \mapsto \hat{x}^t = \{\hat{x}_i^t \in [0,1], i=1,\cdots,N\},$$  
where \( \theta^t \) represents the parameters of the DNN at the current time frame. The network's output layer uses the Sigmoid activation function to ensure bounded outputs. Subsequently, the quantizer converts the continuous vector \( \hat{x}^t \) into \( M_t \) feasible binary candidate offloading actions. To achieve a good balance between exploration and exploitation during algorithm training, we adopt the Noisy Order Preserving (NOP) quantization method \cite{9449944}. This method not only preserves the network's probabilistic tendencies in the output but also generates additional exploratory actions by injecting Gaussian noise, thereby effectively preventing the model from becoming trapped in local optima.


2) Critic Module:  
Unlike traditional Actor Critic algorithms that use model-free DNNs to estimate value functions, this module fully leverages system model information to accurately evaluate the \( M_t \) candidate binary actions generated by the Actor module through analytical solutions to the optimal resource allocation problem. For a given candidate action \( x_j^t \), the Critic module executes the power allocation algorithm designed later in this paper to determine the NOMA transmission power \( \mathbf{P}^t \) and local computation frequency \( \mathbf{f}^t \) under the current offloading combination. Subsequently, the Critic module computes the maximum objective function value for the candidate action. It selects the action with the highest objective value as the final executed action \( \mathbf{x}^t \) for the current time frame. This accurate evaluation mechanism, based on a real optimization model, dramatically enhances the robustness of DRL training and accelerates convergence. To balance system performance and computational complexity, the number of quantized actions \( M_t \) is adaptively adjusted during training. As the DNN policy gradually converges to the optimum, \( M_t \) decreases, significantly reducing the computational latency during the online inference phase.

3) Policy Update Module:  
The Policy Update module employs an experience replay mechanism with a replay memory to update the DNN's parameters. The environmental observation state and the optimal action selected by the Critic module \( (\xi^t, \mathbf{x}^t) \) are stored as a labeled data sample in the memory. To avoid model overfitting, the DNN parameters are updated every \( \delta_T \) time frames. During training, the system randomly samples a batch of samples from the memory. It updates the network parameters \( \theta^t \) using the Adam optimizer, minimizing the average cross-entropy loss. The loss function is defined as:  


\begin{equation*}
\begin{aligned}
LS(\theta^t) = & -\frac{1}{|\mathcal{S}^t|} \sum_{\tau \in \mathcal{S}^t} \left[ (\mathbf{x}^\tau)^\top \log \Pi_{\theta^t}(\xi^\tau) \right. \\
& \left. + (\mathbf{1}-\mathbf{x}^\tau)^\top \log(\mathbf{1}-\Pi_{\theta^t}(\xi^\tau)) \right],
\end{aligned}
\end{equation*}

where \( |\mathcal{S}^t| \) is the size of the sampled batch.

4) Queueing Module:  
The system executes actual computational tasks based on the joint offloading and resource allocation decisions (including NOMA-based transmission power and local computation frequency) determined by the Critic module. Based on the actual amount of data processed and energy consumed in the current time frame, as well as the observed random data arrivals, the Queueing module updates the actual data queues \( \mathbf{Q}(t+1) \) and virtual energy queues \( \mathbf{Y}(t+1) \) of all devices at the beginning of the next time frame using the Lyapunov queue dynamics equation. The updated queue states, combined with the channel gain states of the next frame, form the new environmental input \( \xi^{t+1} \), thereby initiating the next control cycle.

Given the offloading decision \( x^t \) in a single frame subproblem, the set of wireless devices with \( x_i^t=1 \) (i.e., devices performing computation offloading) is denoted as \( \mathcal{M}_1^t \). The complementary set of local computing devices is denoted as \( \mathcal{M}_0^t \). For simplicity and without causing ambiguity, we omit the superscript \( t \) in the subsequent derivations in this section. Given the offloading decision \( \mathbf{x} \), the problem of finding the optimal continuous resource allocation variables can be rewritten as follows:

\begin{subequations}
\begin{align}
& \underset{\mathbf{f}, \mathbf{P}}{\text{maximize}} \quad
\sum_{j\in\mathcal{M}_0} \Bigg\{ \left(Q_j(t) + V c_j\right) \frac{f_j}{\phi} - \nu Y_j(t) \kappa \left(f_j\right)^3 \Bigg\} \nonumber\\
& \quad + \sum_{i\in\mathcal{M}_1} \Bigg\{ \left(Q_i(t) + V c_i\right) \frac{\eta_q W}{v_u} \log_2 \left(1 + \text{SINR}_i \right)\nonumber\\
&\quad - \nu Y_i(t) (P_i + P_{i,\text{qlc}}) \Bigg\} \nonumber\\[6pt]
&\text{subject to} \nonumber \\
&\frac{f_j}{\phi} \leq Q_j(t), \quad 0 \leq f_j \leq f_j^{\max}, \quad \forall j \in \mathcal{M}_0, \\
& \frac{\eta_q W}{v_u} \log_2 \left(1 + \text{SINR}_i\right) \leq Q_i(t),\\
& 0 \leq P_i \leq P_i^{\max}, \quad \forall i \in \mathcal{M}_1.
\end{align}
\end{subequations}





The above optimization problem can be decomposed and solved independently for the wireless devices in the sets $\mathcal{M}_1$ and $\mathcal{M}_0$, respectively. Specifically, for each device $j \in \mathcal{M}_0$ performing local computation, it only needs to solve the following independent optimization subproblem:  
$$\begin{aligned}
&\underset{f_j}{\text{maximize}} \quad \left(Q_j(t) + V c_j\right) f_j / \phi - \nu Y_j(t) \kappa f_j^3 \\
&\text{subject to } \quad 0 \leq f_j \leq \min\left\{ \phi Q_j(t), f_j^{\max} \right\}, \quad \forall j \in \mathcal{M}_0.
\end{aligned}$$

Since the objective function is strictly concave with respect to $f_j$, by taking its first order derivative $\frac{\left(Q_j(t) + Vc_j\right)}{\phi} - 3\nu Y_j(t)\kappa f_j^2 = 0$ and combining it with the boundary constraints, we can directly obtain its closed form optimal solution:  
$$f_j^* = \min\left\{ \sqrt{\frac{\left(Q_j(t) + V c_j\right)}{3\phi \kappa \nu Y_j(t)}}, \min\left\{ \phi Q_j(t), f_j^{\max} \right\} \right\}.$$

Intuitively, when the pending data queue $Q_j(t) + Vc_j$ of the $j$-th wireless device is larger, or the virtual energy queue $Y_j(t)$ is shorter, the device tends to compute at a higher local CPU frequency, and vice versa.

For each offloading device $j \in \mathcal{M}_1$, it needs to solve the following independent optimization subproblem:

\begin{subequations}
\begin{align}
&\underset{\mathbf{P}^t}{\text{maximize}} \quad
\sum_{i\in\mathcal{M}_1} \Bigg\{ \left(Q_i(t) + V c_i\right) \frac{\eta_q W}{v_u} \log_2 \left(1 + \text{SINR}_i^t \right) \notag \\
& \quad - \nu Y_i(t) (P_i^t + P_{i,\text{qlc}}) \Bigg\} \nonumber\\
&\text{subject to} \nonumber \\
& \frac{\eta_q W}{v_u} \log_2 \left(1 + \text{SINR}_i^t\right) \leq Q_i(t), \quad \forall i \in \mathcal{M}_1, \label{const2} \\
& 0 \leq P_i^t \leq P_i^{\max}, \quad \forall i \in \mathcal{M}_1. \label{const3}
\end{align}
\end{subequations}













\section{Power Allocation Algorithm Design and Comparison}\label{sec:Power Allocation}

\subsection{Heuristic Greedy Algorithm}\label{heuristic}
We arrange the channel gains in descending order:
\begin{equation}
h_1^t \geq h_2^t \geq \cdots \geq h_N^t.
\end{equation}

Adopting the most commonly used decoding order strategy, the signal with the strongest channel gain is decoded first. Let
\begin{equation}
\rho_i =\text{SINR}_i= \frac{P_i^t h_i^t}{N_0 + \sum_{j=i+1}^{N} P_j^t h_j^t}, \quad \forall i \in \mathcal{M}_1.
\end{equation}

Then the transmission rate is $r_i = \tfrac{\eta_q W}{v_u} \log_2\left(1 + \rho_i\right)$. We express $P_i^t$ in terms of $\rho_i$. The users in the set $\mathcal{M}_1$ are arranged in ascending order of their indices as $i_1, i_2, \ldots, i_M$ (where $i_1 < i_2 < \cdots < i_M$).

For user $i_M$, there is no interference from other $\mathcal{M}_1$ users:

\begin{equation}
\rho_{i_M} = \frac{P_{i_M}^t h_{i_M}^t}{N_0} \quad \Rightarrow \quad P_{i_M}^t = \frac{\rho_{i_M} N_0}{h_{i_M}^t}.
\end{equation}

Consider the second to last user $i_{M-1}$:

\begin{align*}
\rho_{i_{M-1}} &= \frac{P_{i_{M-1}}^t h_{i_{M-1}}^t}{N_0 + P_{i_M}^t h_{i_M}^t} \\
&= \frac{P_{i_{M-1}}^t h_{i_{M-1}}^t}{N_0 + \frac{\rho_{i_M} N_0}{h_{i_M}^t} \cdot h_{i_M}^t} \\
&= \frac{P_{i_{M-1}}^t h_{i_{M-1}}^t}{N_0 (1 + \rho_{i_M})}.
\end{align*}

Therefore,
\begin{equation}
P_{i_{M-1}}^t = \frac{\rho_{i_{M-1}} N_0 (1 + \rho_{i_M})}{h_{i_{M-1}}^t}.
\end{equation}.

For user $i_k$, the interference term is defined as:

\begin{equation}
I_{i_k} = N_0 + \sum_{j=k+1}^{M} P_{i_j}^t h_{i_j}^t
\label{eq:interference_def}.
\end{equation}.

We propose the following recursive relation.

\begin{proposition}
\label{prop:interference}
For any $k \in \{1, \dots, M\}$, the cumulative interference term for users in set $\mathcal{M}_1$ satisfies
\begin{equation}
    I_{i_k} = N_0 \prod_{j=k+1}^{M}(1 + \rho_{i_j}),
    \label{eq:interference_recursive}
\end{equation}
where the empty product (for $k = M$) is defined as $1$.
\end{proposition}

\begin{IEEEproof}
The proof proceeds by mathematical induction.

\textit{Base case for $k = M$:} According to the definition of the interference term~\eqref{eq:interference_def}, the user with the weakest channel gain $i_M$

is not interfered with by other users in set $\mathcal{M}_1$, therefore
\begin{equation*}
    I_{i_M} = N_0.
\end{equation*}
The empty product on the right-hand side of formula~ equals $1$, so $N_0 \times 1 = N_0$, and the base case holds.

\textit{Inductive hypothesis:} Assume the proposition holds for user $i_{k+1}$, i.e.,
\begin{equation*}
    I_{i_{k+1}} = N_0 \prod_{j=k+2}^{M}(1 + \rho_{i_j}).
\end{equation*}

\textit{Inductive step:} Consider the interference term for user $i_k$:
\begin{equation*}
    I_{i_k} = N_0 + \sum_{j=k+1}^{M} P_{i_j}^t h_{i_j}^t
             = N_0 + P_{i_{k+1}}^t h_{i_{k+1}}^t + \sum_{j=k+2}^{M} P_{i_j}^t h_{i_j}^t.
\end{equation*}
From the definition of $\rho_{i_{k+1}}$, we have $P_{i_{k+1}}^t h_{i_{k+1}}^t = \rho_{i_{k+1}} I_{i_{k+1}}$;
Simultaneously, from the definition of the interference term, we know $\sum_{j=k+2}^{M} P_{i_j}^t h_{i_j}^t = I_{i_{k+1}} - N_0$.
Substituting these two expressions yields
\begin{align*}
    I_{i_k} &= N_0 + \rho_{i_{k+1}} I_{i_{k+1}} + I_{i_{k+1}} - N_0 \\
             &= I_{i_{k+1}}(1 + \rho_{i_{k+1}}).
\end{align*}
Substituting the inductive hypothesis for $I_{i_{k+1}}$, we obtain
\begin{equation*}
    I_{i_k} = N_0 \prod_{j=k+2}^{M}(1+\rho_{i_j}) \cdot (1+\rho_{i_{k+1}})
             = N_0 \prod_{j=k+1}^{M}(1+\rho_{i_j}).
\end{equation*}
The proposition holds for $i_k$, completing the induction.
\end{IEEEproof}

From the definition of $\rho_{i_k}$:

\begin{equation}
\rho_{i_k} = \frac{P_{i_k}^t h_{i_k}^t}{I_{i_k}} \quad \Rightarrow \quad P_{i_k}^t h_{i_k}^t = \rho_{i_k} I_{i_k}.
\end{equation}

Substituting the expression for $I_{i_k}$:

\begin{equation}
P_{i_k}^t h_{i_k}^t = \rho_{i_k} N_0 \prod_{j=k+1}^{M} (1 + \rho_{i_j}).
\end{equation}

We finally obtain:

\begin{equation}
P_{i_k}^t = \frac{\rho_{i_k} N_0}{h_{i_k}^t} \prod_{j=k+1}^{M} (1 + \rho_{i_j}).
\end{equation}

Rate constraint

\begin{align*}
\frac{\eta_qW}{v_u} \log_2(1 + \rho_i) \leq Q_i(t) &\Rightarrow \log_2(1 + \rho_i) \leq \frac{v_u Q_i(t)}{\eta_qW} \\
&\Rightarrow 1 + \rho_i \leq 2^{\frac{v_u Q_i(t)}{\eta_qW}} \\
&\Rightarrow \rho_i \leq 2^{\frac{v_u Q_i(t)}{\eta_qW}} - 1 = \Gamma_i.
\end{align*}

Power constraint:

\begin{align}
0 \leq P_i^t \leq P_i^{\max} &\Rightarrow 0 \leq \frac{\rho_i N_0}{h_i^t} \prod_{j=k+1}^{M} (1 + \rho_{i_j}) \leq P_i^{\max}.
\end{align}

Since $\rho_i \geq 0$, the left inequality naturally holds. The right inequality yields:

\begin{equation}
\rho_i \leq \frac{h_i^t P_i^{\max}}{N_0 \prod_{j=k+1}^{M} (1 + \rho_{i_j})} = U_i.
\end{equation}

Let $a_i = \eta_q\left(Q_i(t) + V c_i\right)$ and $b_i = \nu Y_i(t)$.

The objective function term becomes:

\begin{align*}
&a_i \frac{W}{v_u} \log_2(1 + \rho_i) - b_i P_i^t\\
&= a_i \frac{W}{v_u} \log_2(1 + \rho_i) - b_i \frac{\rho_i N_0}{h_i^t} \prod_{j=k+1}^{M} (1 + \rho_{i_j}).
\end{align*}

Ignoring the constant $P_{i,qlc}$ since it does not affect the optimization, let

\begin{equation}
C_{k+1} = \prod_{j=k+1}^{M} (1 + \rho_{i_j}).
\end{equation}

Then the objective term for user $i_k$ is:

\begin{equation}
a_{i_k} \frac{W}{v_u} \log_2(1 + \rho_{i_k}) - b_{i_k} \frac{\rho_{i_k} N_0}{h_{i_k}^t} C_{k+1}.
\end{equation}

Due to the non-convexity and strong coupling of the global power allocation problem, increasing the power of user $i_k$ will increase the interference experienced by all users with better channel conditions. To obtain a computationally efficient solution, we propose a greedy backward induction strategy. In this approach, we decompose the global optimization into a series of local subproblems. For each user $i_k$, we maximize its contribution to the objective function while satisfying its specific constraints, ignoring the marginal interference cost imposed on previous users. The problem is solved in reverse order, starting from the last decoded user ($i_M$) and proceeding to the first user ($i_1$). For a general user $i_k$, assuming that the power allocation variables for subsequent users $\rho_{i_{k+1}}, \dots, \rho_{i_M}$ have been determined, the cumulative interference term $C_{k+1}$ is fixed:

\begin{equation}
C_{k+1} = \prod_{j=k+1}^{M} (1 + \rho_{i_j}^*).
\end{equation}
For the last user $i_M$, we define $C_{M+1} = 1$. The local optimization problem for user $i_k$ can be formulated as:

\begin{align}
\underset{\rho_{i_k}}{\text{maximize}} \quad & \mathcal{F}_{k}(\rho_{i_k}) = a_{i_k} \frac{W}{v_u} \log_2(1 + \rho_{i_k}) - b_{i_k} \frac{\rho_{i_k} N_0}{h_{i_k}^t} C_{k+1} \nonumber \\
\text{subject to} \quad & 0 \leq \rho_{i_k} \leq \min\left( \Gamma_{i_k}, \frac{h_{i_k}^t P_{i_k}^{\max}}{N_0 C_{k+1}} \right).
\end{align}

Analyzing the convexity of the local objective function $\mathcal{F}_{k}(\rho_{i_k})$. Taking the first order derivative with respect to $\rho_{i_k}$:
\begin{align}
\frac{\partial \mathcal{F}_{k}}{\partial \rho_{i_k}} &= a_{i_k} \frac{W}{v_u} \cdot \frac{1}{\ln 2} \cdot \frac{1}{1+\rho_{i_k}} - b_{i_k} \frac{N_0 C_{k+1}}{h_{i_k}^t} \nonumber \\
&= \frac{a_{i_k} W}{v_u \ln 2 (1+\rho_{i_k})} - \frac{b_{i_k} N_0 C_{k+1}}{h_{i_k}^t}.
\end{align}

Setting the derivative to zero $\frac{\partial \mathcal{F}_{k}}{\partial \rho_{i_k}} = 0$, we obtain the unconstrained stationary point $\rho_{i_k}^0$:

\begin{equation}
\frac{a_{i_k} W}{v_u \ln 2 (1+\rho_{i_k}^0)} = \frac{b_{i_k} N_0 C_{k+1}}{h_{i_k}^t}.
\end{equation}

Solving for $\rho_{i_k}^0$:
\begin{equation}
\rho_{i_k}^0 = \frac{a_{i_k} W h_{i_k}^t}{v_u \ln 2 \cdot b_{i_k} N_0 C_{k+1}} - 1.
\end{equation}

Since the second order derivative $\frac{\partial^2 \mathcal{F}_{k}}{\partial \rho_{i_k}^2} = - \frac{a_{i_k} W}{v_u \ln 2 (1+\rho_{i_k})^2} < 0$, the local objective function is strictly concave. Therefore, the optimal solution to the local subproblem lies at a stationary point or on the boundary of the feasible set. The closed-form heuristic solution for user $i_k$ is given by:
\begin{equation}
\rho_{i_k}^* = \min\left( \max(0, \rho_{i_k}^0), \min\left( \Gamma_{i_k}, \frac{h_{i_k}^t P_{i_k}^{\max}}{N_0 C_{k+1}} \right) \right).
\end{equation}

After determining $\rho_{i_k}^*$, the cumulative interference term is updated for the next iteration (user $i_{k-1}$):
\begin{equation}
C_k = C_{k+1} (1 + \rho_{i_k}^*) = \prod_{j=k}^{M} (1 + \rho_{i_j}^*).
\end{equation}

For the special case where $b_{i_k} = 0$, if $Y_i(t) = 0$, then $b_{i_k} = 0$, and the objective function becomes:
\begin{equation}
\text{Term}_{i_k} = a_{i_k} \frac{W}{v_u} \log_2(1 + \rho_{i_k}).
\end{equation}

This is a monotonically increasing function of $\rho_{i_k}$, so the optimal solution takes the maximum value:
\begin{equation}
\rho_{i_k}^* = \min\left( \Gamma_{i_k}, \frac{h_{i_k}^t P_{i_k}^{\max}}{N_0 C_{k+1}} \right).
\end{equation}

\subsection{SCA Benchmark Algorithm}

To evaluate the performance of the low-complexity heuristic algorithm proposed in Section \ref{heuristic}, we use an iterative SCA-based \cite{razaviyayn2013unified} as a performance benchmark. The SCA algorithm guarantees convergence to a KKT point of the original nonconvex problem by iteratively solving a series of convex approximate subproblems. Consider the optimization problem in Section \ref{sec:AC}. To apply the SCA method, we first rewrite the objective function as a Difference of Convex (DC) functions.

Let $\alpha_i = (Q_i(t) + Vc_i)\frac{\eta_qW}{v_u}$ and $\beta_i = \nu Y_i(t)$. Ignoring the constant term $P_{i,\text{qlc}}$ (which does not affect the optimization), the objective function can be expressed as:

\begin{align}
\mathcal{F}(\mathbf{P}^t) 
&= \sum_{i \in \mathcal{M}_1} \left[ \alpha_i \log_2(1 + \text{SINR}_i^t) - \beta_i P_i^t \right] \nonumber \\
&= \sum_{i \in \mathcal{M}_1} \alpha_i \left[ \underbrace{\log_2\left(N_0 + \sum_{j=i}^{N} P_j^t h_j^t\right)}_{g_i(\mathbf{P}^t)} \right. \nonumber \\
&\quad \left. - \underbrace{\log_2\left(N_0 + \sum_{j=i+1}^{N} P_j^t h_j^t\right)}_{l_i(\mathbf{P}^t)} \right] 
- \sum_{i \in \mathcal{M}_1} \beta_i P_i^t.
\label{eq:dc_decomposition}
\end{align}

Where $g_i(\mathbf{P}^t)$ and $l_i(\mathbf{P}^t)$ are both concave functions with respect to $\mathbf{P}^t$ (because the linear function inside the logarithm is concave, and the logarithmic function is an increasing concave function). Therefore, the original maximization problem is equivalent to maximizing a "concave function minus a concave function", which is a typical DC (Difference of Convex) programming problem. Meanwhile, the rate constraint (\ref{const2}) can be equivalently transformed into a linear inequality:
\begin{equation}
P_i^t h_i^t \leq \Gamma_i \left( N_0 + \sum_{j=i+1}^{N} P_j^t h_j^t \right), \quad \forall i \in \mathcal{M}_1
\label{eq:linear_constraint}.\,
\end{equation}
where $\Gamma_i = 2^{\frac{v_u Q_i(t)}{\eta_qW}} - 1$. Constraint (\ref{eq:linear_constraint}) is linear with respect to $\mathbf{P}^t$ and inherently satisfies convexity. The non convexity primarily originates from the subtractive term $-\alpha_i l_i(\mathbf{P}^t)$ in the objective function (\ref{eq:dc_decomposition}). Since $l_i(\mathbf{P}^t)$ is a concave function, its first-order Taylor expansion provides a global upper bound for this function.

At the $k$-th iteration, given a local point $\mathbf{P}^{(k)} = [P_1^{(k)}, \dots, P_N^{(k)}]$, we perform a first order Taylor expansion on $l_i(\mathbf{P}^t)$:
\begin{align}
l_i(\mathbf{P}^t) &\leq l_i(\mathbf{P}^{(k)}) + \nabla l_i(\mathbf{P}^{(k)})^T (\mathbf{P}^t - \mathbf{P}^{(k)}) \nonumber \\
&= \tilde{l}_i(\mathbf{P}^t; \mathbf{P}^{(k)}),
\label{eq:taylor_approximation}
\end{align}

where the $j$-th component of the gradient $\nabla l_i(\mathbf{P}^{(k)})$ is:
\begin{equation}
\frac{\partial l_i}{\partial P_j^t}\bigg|_{\mathbf{P}^{(k)}} = 
\begin{cases} 
\frac{h_j^t}{\left(N_0 + \sum_{m=i+1}^{N} P_m^{(k)} h_m^t\right) \ln 2}, & \text{if } j > i \\
0, & \text{otherwise}
\end{cases}.
\label{eq:gradient}
\end{equation}

By replacing $l_i(\mathbf{P}^t)$ with its linear upper bound $\tilde{l}_i(\mathbf{P}^t; \mathbf{P}^{(k)})$, we obtain the convex approximate subproblem at the $k$-th iteration:
\begin{subequations}
\begin{align}
\underset{\mathbf{P}^t}{\text{maximize}} \quad & \sum_{i \in \mathcal{M}_1} \left[ \alpha_i \left( g_i(\mathbf{P}^t) - \tilde{l}_i(\mathbf{P}^t; \mathbf{P}^{(k)}) \right) - \beta_i P_i^t \right] \\
\text{subject to} \quad & P_i^t h_i^t - \Gamma_i \sum_{j=i+1}^{N} P_j^t h_j^t \leq \Gamma_i N_0, \quad \forall i \in \mathcal{M}_1, \label{eq:sca_sub_const1} \\
& 0 \leq P_i^t \leq P_i^{\max}, \quad \forall i \in \mathcal{M}_1. \label{eq:sca_sub_const2}
\end{align}
\end{subequations}

Subproblem (\ref{eq:sca_sub_const1})-(\ref{eq:sca_sub_const2}) is a standard convex optimization problem \cite{bertsekas1997nonlinear}, where the objective function is concave (a concave function minus a linear function remains concave), and the constraints are all linear. This problem can be efficiently solved using interior-point methods . The detailed algorithm is shown in Algorithm \ref{alg:SCA_benchmark}.


\begin{algorithm}[H]
\caption{SCA-based Power Allocation Benchmark Algorithm}\label{alg:SCA_benchmark}
\begin{algorithmic}[1]
\REQUIRE $\mathcal{M}_1$, $Q_i(t)$, $Y_i(t)$, $h_i^t$, $V$, $\nu$, $W$, $v_u$, $c_i$, $N_0$, $P_i^{\max}$, $\eta_q$
\ENSURE Optimal power allocation $\mathbf{P}^{t*}$
\STATE \textbf{Initialization:} iteration index $k \gets 0$, tolerance $\epsilon \gets 10^{-4}$, maximum number of iterations $K_{\max} \gets 100$
\STATE Generate an initial feasible point $\mathbf{P}^{(0)}$ (e.g., uniform power allocation)
\REPEAT
\STATE Compute $\Gamma_i = 2^{\frac{v_u Q_i(t)}{\eta_qW}} - 1$, $\forall i \in \mathcal{M}_1$
\STATE Compute $\alpha_i = (Q_i(t) + Vc_i)\frac{\eta_qW}{v_u}$, $\beta_i = \nu Y_i(t)$, $\forall i \in \mathcal{M}_1$
\STATE Compute the gradient $\nabla l_i(\mathbf{P}^{(k)})$ according to (\ref{eq:gradient})
\STATE Construct the convex subproblem (\ref{eq:sca_sub_const1})-(\ref{eq:sca_sub_const2})
\STATE Solve the subproblem using a convex optimization solver (e.g., CVX) to obtain $\mathbf{P}^{(k+1)}$
\STATE Compute the change in objective function value: $\Delta = |\mathcal{F}(\mathbf{P}^{(k+1)}) - \mathcal{F}(\mathbf{P}^{(k)})|$
\STATE $k \gets k + 1$
\UNTIL{$\Delta \le \epsilon$ or $k \ge K_{\max}$}
\STATE $\mathbf{P}^{t*} \gets \mathbf{P}^{(k)}$
\STATE \textbf{Return} $\mathbf{P}^{t*}$
\end{algorithmic}
\end{algorithm}






\subsection{Computational Complexity Analysis}

Let the cardinality of the set $\mathcal{M}_1$ of devices performing computation offloading be $M = |\mathcal{M}_1| \leq N$.
The main body of the algorithm consists of two sequential traversal loops: a backward iteration phase (from user $i_M$ to $i_1$) and a forward power computation phase (from user $i_M$ to $i_1$).
Each loop requires $M$ iterations, and each iteration involves only a constant number of floating-point operations. Computing the stationary point $\rho_{i_k}^0$, updating the cumulative interference term $C_k$, etc.
Therefore, for a given offloading decision $\mathbf{x}^t$, the computational complexity of Algorithm~1 is $\mathcal{O}(N)$.

Algorithm~\ref{alg:SCA_benchmark} approximates the KKT point of the original nonconvex problem by iteratively solving convex approximate subproblems.
In each SCA iteration, a convex quadratic programming (QP) problem with $|\mathcal{M}_1|$ power variables as optimization variables needs to be constructed and solved, and the number of constraints is $\mathcal{O}(N)$.
The per iteration complexity of solving this QP subproblem using interior point methods is $\mathcal{O}(N^{3.5})$ .
Assuming the number of iterations required for SCA convergence is $K_{\mathrm{iter}}$ (bounded above by $K_{\max}$), the total complexity of Algorithm~2 is $\mathcal{O}(K_{\mathrm{iter}} \cdot N^{3.5})$.
Since $K_{\mathrm{iter}}$ is typically on the order of tens, the computational cost of SCA is significantly higher than that of Algorithm~1, especially when the number of devices $N$ is large.

\section{EVALUATION}\label{sec:evaluation}

\subsection{Parameter Settings}

Simulations are implemented in Python (PyTorch for DNN, CVXPY for SCA). We consider $N=10$ WDs uniformly distributed between 120\,m and 255\,m from the ES. Channel gains follow a Rician fading model with a standard path loss ($d_e=3$). Task arrivals are exponentially distributed with mean $\lambda_i = 3$ Mbps. The DNN actor uses a fully connected architecture $[3N, 256, 128, N]$. Priority weights are set as $c_i=1.5$ for odd-indexed devices and $c_i=1$ otherwise. Other key parameters are listed in Table~\ref{tab:sim_params}.

\begin{table}[htbp]
  \centering
  \caption{Simulation Parameters.}
  \label{tab:sim_params}
  \begin{tabular}{cccc}
    \toprule
    \textbf{Parameter} & \textbf{Value} & \textbf{Parameter} & \textbf{Value} \\
    \midrule
    $W$ & 2 MHz & $f_{\text{max}}$ & 0.3 GHz \\
    $v_{\text{u}}$ & 1.1 & $P_{\text{max}}$ & 0.1 W \\
    memory size $q$ & 1024 & $\phi$ & 100 cycles/bit \\
    $\Delta T$ & 20 & $\Delta M$ & 32 \\
    $\gamma_i$ & 0.08 W & $\nu$ & 60 \\
    $V$ & 20 & $\lambda_i$ & 3 Mbps \\
    \bottomrule
  \end{tabular}
\end{table}












\subsection{Basic Performance Comparison}

\begin{figure}[htbp]
    \centering
    \includegraphics[width=0.45\textwidth]{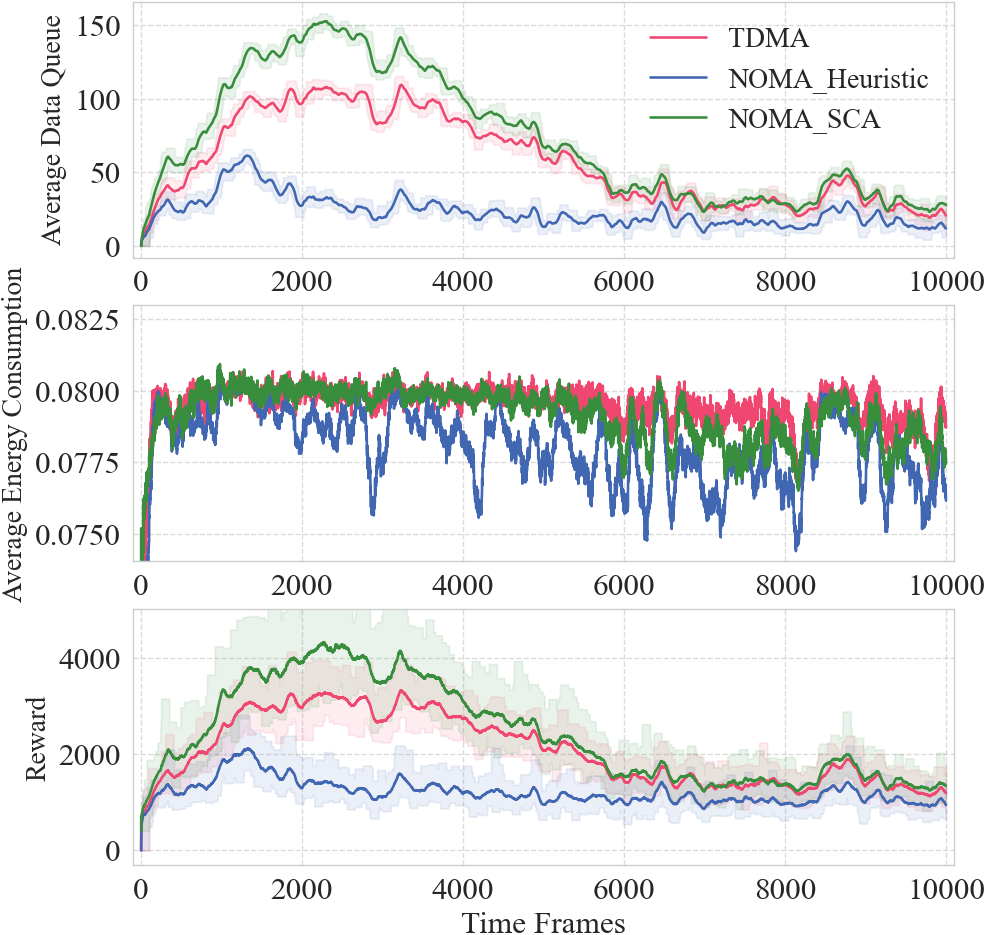} 
    \caption{Convergence performance comparison of different schemes when $\lambda_i$ is 3. From top to bottom: data queue length; system power consumption; system reward.}
    \label{fig:single}
\end{figure}


As shown in Figure \ref{fig:single}, over 10,000 time frames, we compare three algorithm schemes: the TDMA scheme, the $\text{NOMA\_Heuristic}$ scheme, and the $\text{NOMA\_SCA}$ scheme

\textbf{Data Queue Stability}: The $\text{NOMA\_Heuristic}$ scheme demonstrates the best queue stability, with the lowest average data queue length and the fastest convergence to a very low steady state value. The $\text{NOMA\_SCA}$ scheme, in its initial phase, slows down the DNN convergence due to the enormous computational overhead of each iteration, leading to a peak in queue backlog (approximately 150 bits). However, after sufficient training, it can fall back to a level comparable to TDMA.

\textbf{Long-term Energy Constraint}: The average energy consumption of all three schemes closely adheres to the preset threshold of $\gamma_i = 0.08$ W, which validates the effectiveness of the proposed Lyapunov-based virtual energy queue mechanism. Among them, the $\text{NOMA\_Heuristic}$ scheme exhibits slightly lower energy consumption than the other two schemes due to its conservative power strategy that prioritizes emptying the queues.

\textbf{System Reward Value}: SCA indeed outperforms Heuristic in terms of per-frame reward, which is consistent with SCA converging to a KKT point. However, the core advantage of Heuristic lies not in maximizing the per-frame reward, but in its O(N) low latency, which enables the DNN to converge faster and achieve a more stable queue state over the long term.

\subsection{Analysis of System Load Impact}
This subsection evaluates the impact of the task arrival rate $\lambda_i$ on system performance by comparing and analyzing the three algorithms TDMA, NOMA, and SCA under two scenarios: low load with $\lambda \in [1.9, 2.9]$ as shown in Figure \ref{fig:single0}, and high load with $\lambda \in [2.9, 3.9]$ as shown in Figure \ref{fig:single1}.

\textbf{Average Data Queue Length} exhibits a nonlinear increasing trend as the arrival rate $\lambda_i$ grows:
During the low-load phase, all three algorithms can effectively handle arriving tasks, with queue lengths remaining low at below 40 Mb. Among them, the NOMA algorithm benefits from the high spectral efficiency of non-orthogonal multiple access NOMA, consistently maintaining the lowest queue length.
In the high-load phase when $\lambda_i$ exceeds 3.1, the queues of all three algorithms exhibit uncontrolled growth.

\textbf{Average Energy Consumption}: The energy consumption of all algorithms increases with the load, eventually approaching the preset energy threshold of 0.08 watts. In the low load region, at the same load, NOMA's average energy consumption is consistently lower than that of TDMA and SCA. This is because NOMA allows users to transmit in parallel on the same subchannel, reducing the energy per bit of data. When the load $\lambda_i$ exceeds 3.2, the energy consumption curves of all three algorithms surpass the threshold.

\textbf{Average Weighted Sum Computation Rate} reflects the upper limit of system processing capability:
In the low load region, NOMA\_SCA performs very well, decreasing slowly from 40 Mbps, while TDMA and NOMA\_Heuristic are almost identical, increasing slowly from 25 Mbps. In the high load region, NOMA can maintain around 40 Mbps, significantly outperforming the other two algorithms. The simulation results consistently demonstrate that the NOMA algorithm achieves lower queuing delay, higher computational throughput, and better energy efficiency across different system loads, demonstrating its excellent robustness in dynamic I environments.






\begin{figure}[htbp]
    \centering
    \includegraphics[width=0.45\textwidth]{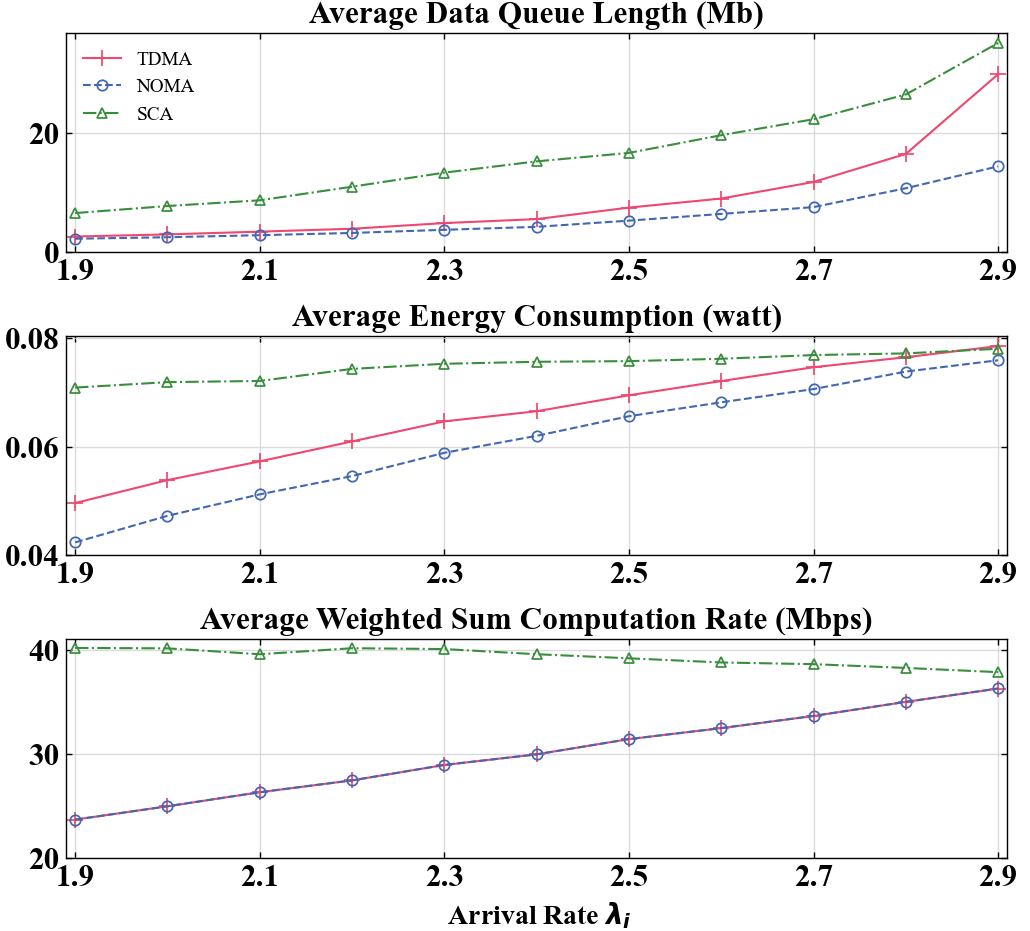} 
    \caption{Performance comparison under different $\lambda_i$ values at low load.}
    \label{fig:single0}
\end{figure}

\begin{figure}[htbp]
    \centering
    \includegraphics[width=0.45\textwidth]{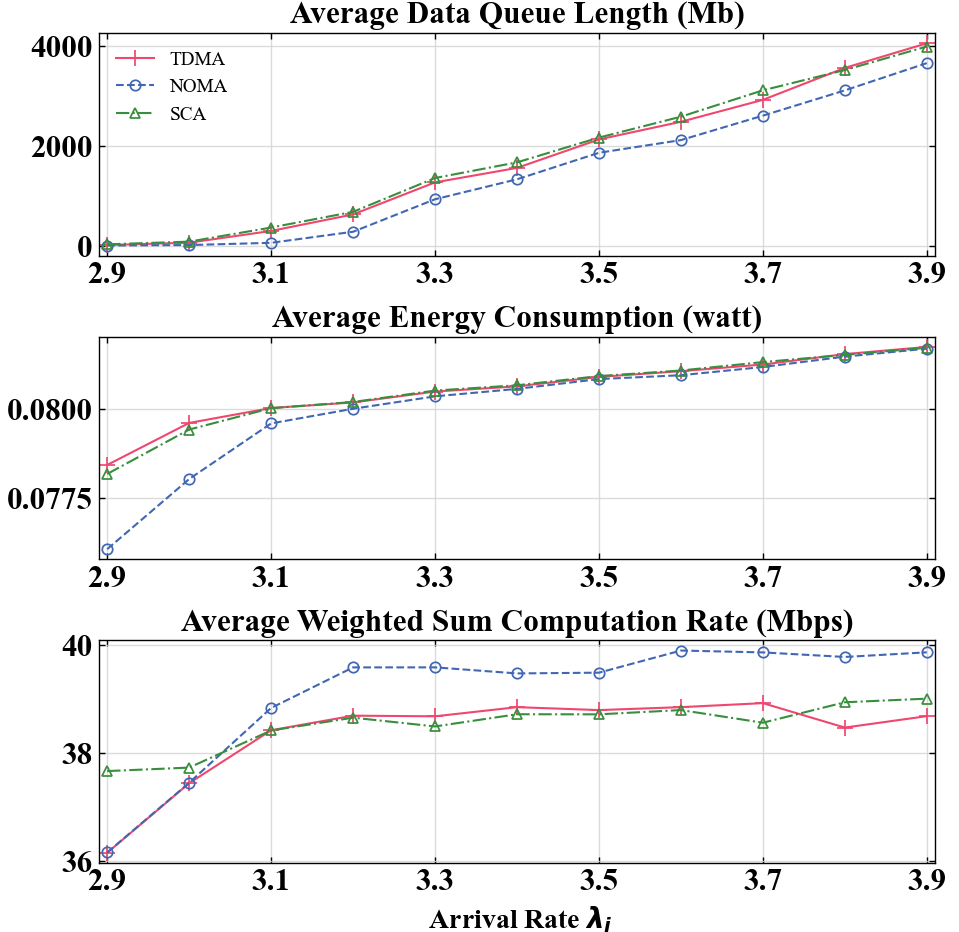} 
    \caption{Performance comparison under different $\lambda_i$ values at high load.}
    \label{fig:single1}
\end{figure}

\subsection{Algorithm Execution Time and Computational Complexity}

In Figure \ref{fig:single2}, with the number of users $N$ increasing from 10 to 35, each configuration is independently repeated 30 times, and the average execution time and standard deviation of each scheme are recorded. The logarithmic scale plot shows the trend in execution time for the three schemes as the number of users increases, with the shaded areas representing the corresponding one-standard-deviation confidence intervals.
The NOMA\_Heuristic scheme achieves the lowest execution time across all tested scales and exhibits the most gradual growth trend. When $N$ increases from 10 to 35, its execution time only grows from approximately 75 $\mu$s to approximately 175 $\mu$s, an increase of about 2.3 times, which is highly consistent with the theoretical $\mathcal{O}(N)$ complexity prediction. For the NOMA\_SCA scheme, when $N = 10$, it is 2534 $\mu$s, and when $N = 35$, it exceeds 8000 $\mu$s, with an overall increase of about 3.2 times, consistent with the theoretical $\mathcal{O}(K_{\text{iter}} \cdot N^{3.5})$ complexity. At $N = 35$, the execution time of NOMA SCA is approximately 46 times that of NOMA\_Heuristic, indicating that as the network scale expands, iterative approximation methods based on convex optimization solvers will face significant latency bottlenecks, making it challenging to meet the real-time requirements for millisecond-level online decision making in fast fading channels. Specific data are shown in Table \ref{tab:comparison}.

\begin{figure}[htbp]
    \centering
    \includegraphics[width=0.45\textwidth]{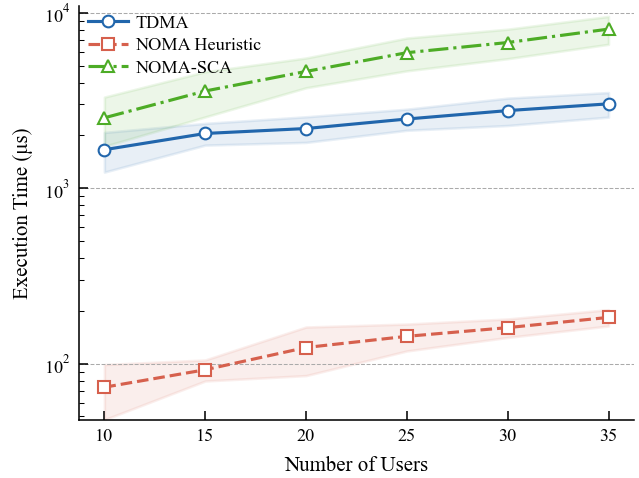} 
    \caption{Comparison of execution time for three algorithms.}
    \label{fig:single2}
\end{figure}

\begin{table}[htbp]
\centering
\caption{Time comparison of three algorithms (microseconds).}
\label{tab:comparison}
\begin{tabular}{l r r r}
\toprule
N & TDMA & NOMA\_Heuristic & NOMA\_SCA \\
\midrule
10 & 1618.26 & 73.86 & 2343.63 \\
15 & 1831.25 & 100.22 & 3291.03 \\
20 & 1994.97 & 116.20 & 4206.43 \\
25 & 2234.49 & 135.66 & 5348.66 \\
30 & 2479.60 & 163.62 & 6031.52 \\
35 & 2692.38 & 179.22 & 7210.54 \\
\bottomrule
\end{tabular}
\end{table}

\subsection{Impact of the Quantum Security Module on Offloading Decisions and System Power Consumption}

In Figure \ref{fig:single3}, with a fixed task arrival rate $\lambda_i = 3$ Mbps, $P_{i,qlc}$ is gradually increased from 0.00 W to 0.10 W to test the average offloading ratio and power consumption after the convergence of the NOMA\_Heuristic scheme. The test results show that as $P_{i,qlc}$ increases, the offloading ratio decreases strictly monotonically, dropping from 25.6\% at 0 W to 16.9\% at 0.10 W, a cumulative decrease of 8.7 percentage points, corresponding to a relative reduction of approximately 34\%. This phenomenon arises because $P_{i,qlc}$, as the fixed circuit power consumption in the offloading decision, directly increases the energy penalty term $\nu Y_i(t) e_i^t$ in the Lyapunov DPP function. This raises the cost of offloading energy consumption, thereby prompting the DNN action network to favor local computation. The above results indicate that the fixed power consumption of the quantum security module has a significant impact on the offloading strategy, and the system design phase requires strategy calibration based on $P_{i,qlc}$.

\begin{figure}[htbp]
    \centering
    \includegraphics[width=0.45\textwidth]{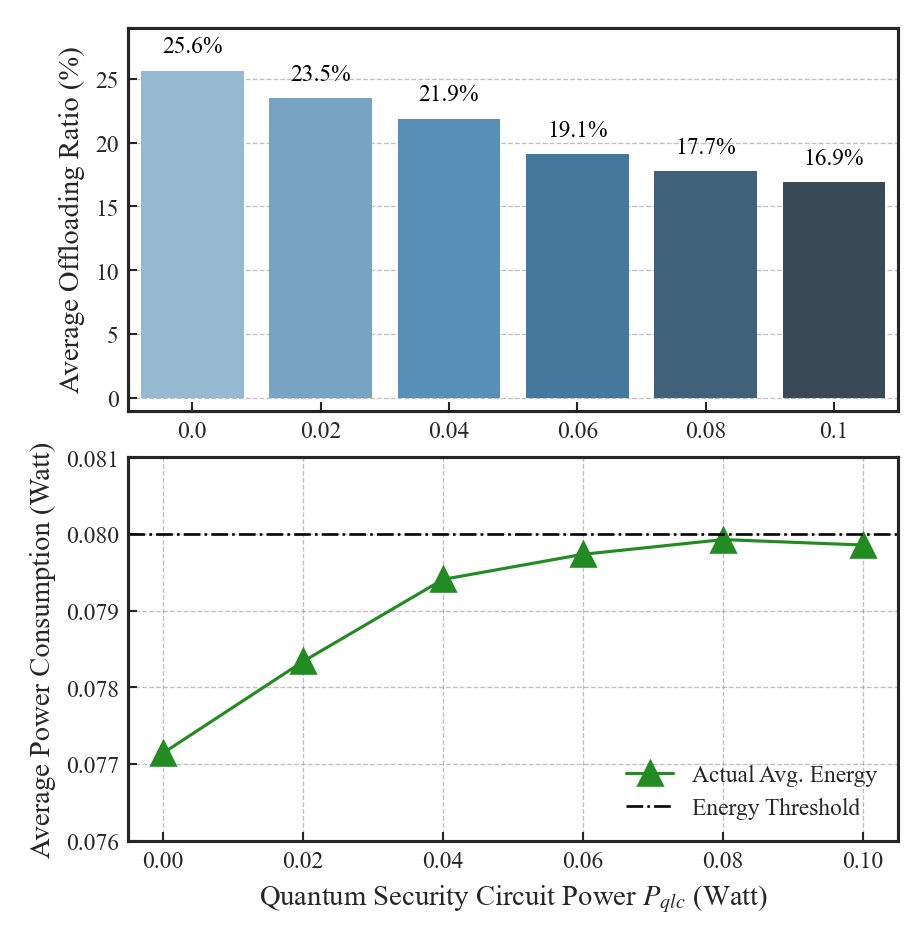} 
    \caption{Impact of PQC power consumption on offloading ratio and system power consumption.}
    \label{fig:single3}
\end{figure}
As shown in the lower part of Figure \ref{fig:single3}, the average system power consumption increases monotonically with $P_{i,qlc}$, rising from approximately 0.0771 W. It approaches the energy consumption threshold $\gamma_i = 0.08$ W when $P_{qlc} \approx 0.08$ W, after which the curve tends to saturate. Throughout the tested range, the average power consumption never exceeds $\gamma_i$, verifying that the Lyapunov virtual energy queue mechanism can effectively maintain the long-term power constraint even under the disturbance of additional fixed power consumption introduced by the quantum security module.

\subsection{Ablation Experiments}

Through systematic ablation experiments, we quantitatively analyze the influence patterns of two key hyperparameters in the Lyapunov framework, the trade-off $V$ and the virtual energy queue scaling factor $\nu$, on system performance. 

\textbf{Analysis of the Impact of Parameter $V$}: $V$ is the core trade-off in the DPP expression that balances queue delay and computation rate. Its value determines the system's emphasis between queue stability and maximizing computational throughput. Fixing $\nu=60$, $\eta_q=1.0$, $P_{\text{qlc}}=0$, comparative experiments are conducted for $V \in \{1, 50, 100, 500, 1000\}$, with results shown in Figure \ref{fig:ab_nu} and Table \ref{tab:effect_nu}.

Regarding the average data queue length, an increase in $V$ leads to significant queue backlog: when $V \leq 100$, the data queue stabilizes at $\leq 25.6\text{ Mb}$, converging within about 3000 frames; however, when $V=500$ and $V=1000$, the queue length rises to 71.9 Mb and 150.5 Mb, respectively, and does not fully converge within 10000 frames, indicating that excessively large $V$ impairs queue stability. This aligns with DPP theory: the larger the $V$, the more the system tends to maximize the computation rate while reducing the queue penalty, leading to queue growth.

Regarding the average energy queue length, an increase in $V$ causes it to grow synchronously, from 4.84 mJ at $V=1$ to 174.4 mJ at $V=1000$, reflecting that under high $V$, the system tolerates higher energy consumption fluctuations to enhance the rate. However, the average power consumption for all configurations converges near the energy consumption threshold $\gamma_i=0.08$ W, verifying that the virtual energy queue mechanism effectively maintains the long-term power constraint.

Regarding the weighted computation rate, the impact of $V$ is limited: from $V=1$ to 1000, the weighted rate only increases slightly from 37.35 Mbps to 37.65 Mbps, an increase of less than 1\%. Further increasing $V$ yields limited rate gains but significantly worsens queue delay. In summary, this paper selects $V=20$ as the default configuration to balance queue stability and computational efficiency.

\begin{figure}
    \centering
    \includegraphics[width=0.45\textwidth]{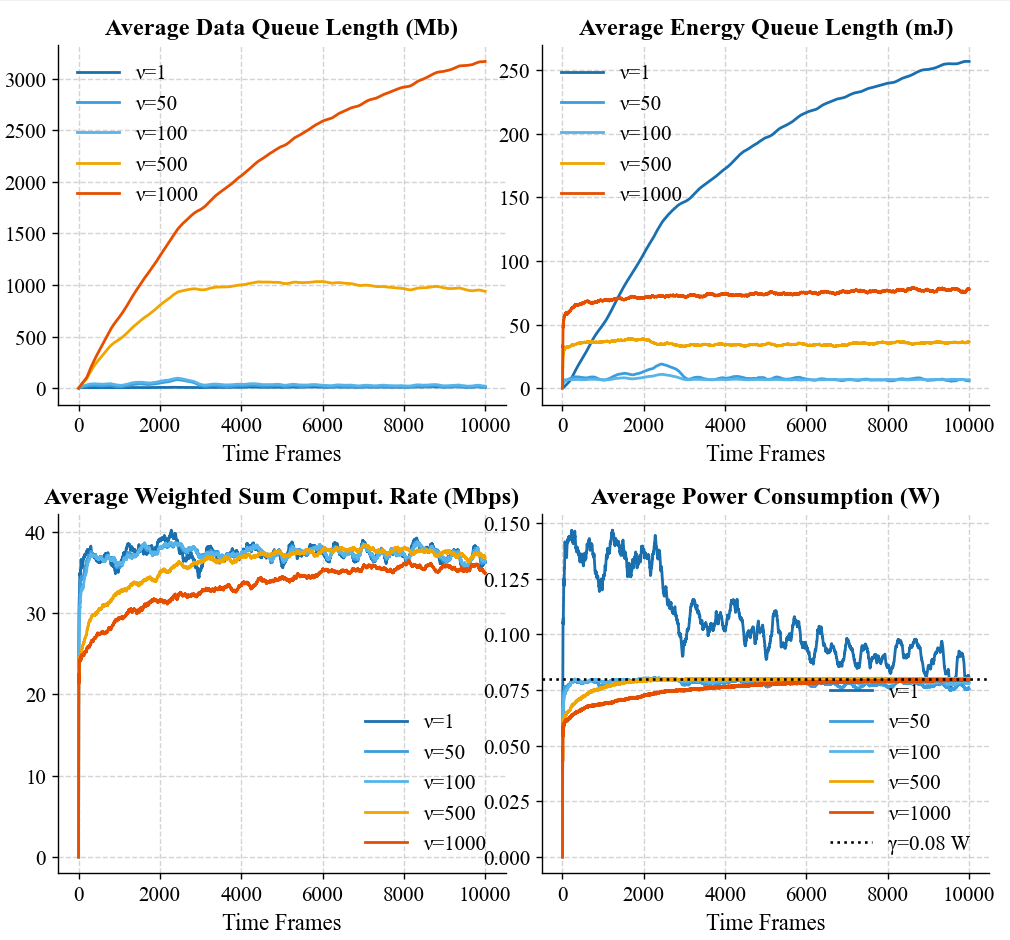} 
    \caption{Evolution curves of key system performance metrics over time frames under different weight parameters $\nu$.}
    \label{fig:ab_nu}
\end{figure}

\begin{table}[htbp]
    \centering
    \caption{Analysis of the impact of control parameter V on system queue length, computation rate, and power consumption.}
    \label{tab:effect_V}
    \begin{tabular}{c c c c c}
        \toprule
        $V$ & \makecell{Data Queue \\ (Mb)} & \makecell{Energy Queue \\ (mJ)} & \makecell{Weighted Rate \\ (Mbps)} & \makecell{Power \\ (W)} \\
        \midrule
        1    & 24.632  & 4.840   & 37.35 & 0.0793 \\
        50   & 18.557  & 10.786  & 37.35 & 0.0779 \\
        100  & 25.585  & 19.218  & 37.36 & 0.0790 \\
        500  & 71.915  & 85.294  & 37.49 & 0.0799 \\
        1000 & 150.488 & 174.369 & 37.65 & 0.0799 \\
        \bottomrule
    \end{tabular}
\end{table}

The parameter $\nu$ scales the energy queue update, determining the penalty strength for energy constraint violations in the DPP objective. With $V=20$, experiments were conducted for $\nu \in \{1, 50, 100, 500, 1000\}$, with results shown in Figure \ref{fig:ab_V} and Table \ref{tab:effect_V}.

For $\nu=1$, the energy penalty is too weak. The virtual energy queue grows monotonically without converging over 10,000 frames, reaching a steady-state mean of 234.1 mJ. The average power consumption is 0.0919 W, exceeding the threshold $\gamma_i=0.08$ W, indicating a "soft failure" where energy feasibility is not guaranteed.

At $\nu=50$ and $100$, the system performs optimally. The energy queue converges rapidly with steady-state means of 6.86 mJ and 6.68 mJ. The data queue remains stable, average power meets the threshold, and the weighted computation rate stays high at 37.35 Mbps. Here, the penalty effectively balances energy constraints and throughput.

When $\nu$ increases to $500$ and $1000$, an excessive energy penalty becomes counterproductive. The data queue experiences a severe backlog and grows unbounded, exceeding 10,000 frames. The DPP objective, dominated by the energy term, forces the system to suppress transmission power overly, drastically reducing the offloading rate and destabilizing the queue. At $\nu=1000$, the weighted computation rate drops to 35.31 Mbps,a 5.5\% decrease from the optimal setting.

These results reveal an optimal window for $\nu$: too small fails to enforce energy constraints, too large compromises queue stability and throughput. Balancing data queue stability, energy feasibility, and system throughput, this paper adopts $\nu=60$ as the default configuration.

\begin{figure}
    \centering
    \includegraphics[width=0.45\textwidth]{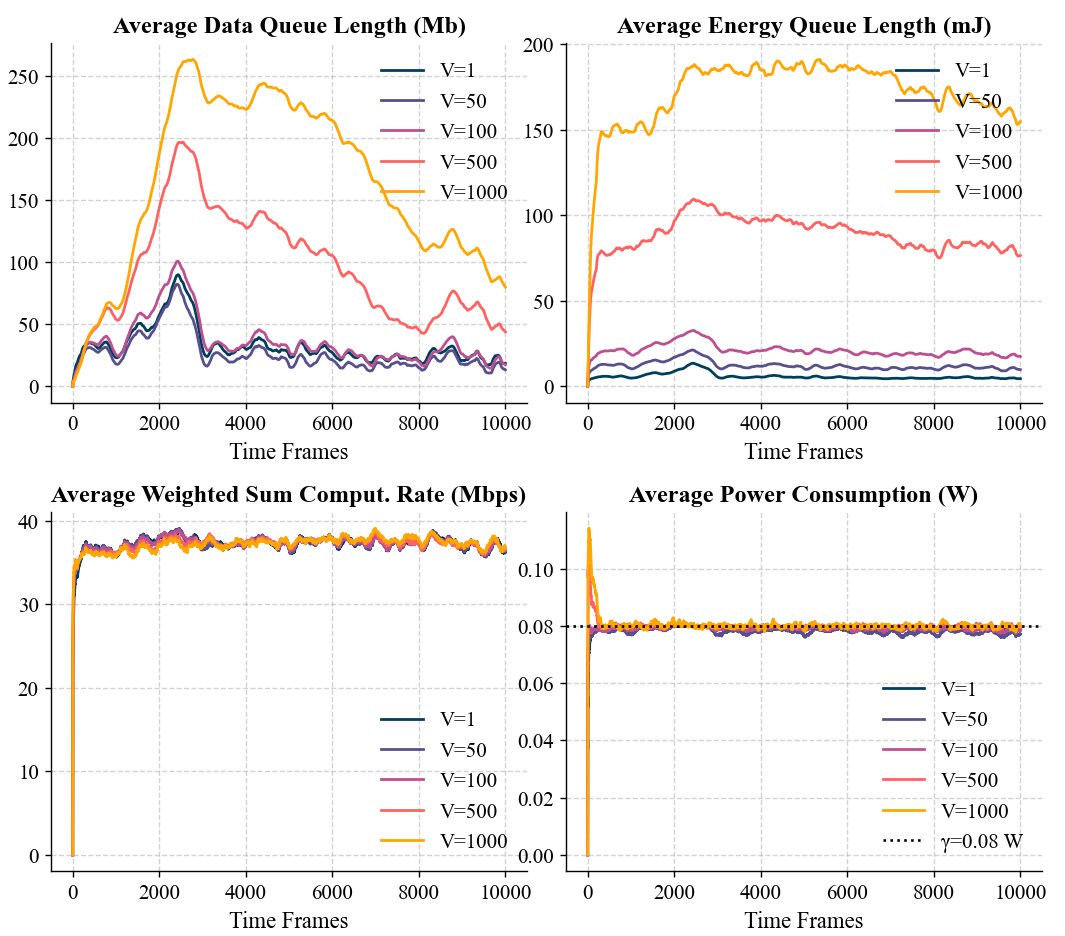} 
    \caption{Evolution curves of key system performance indicators over time frames under different weight parameters $V$.}
    \label{fig:ab_V}
\end{figure}

\begin{table}
    \centering
    \caption{Analysis of the influence of control parameter $\nu$  on system queue length, computing rate and power consumption}
    \label{tab:effect_nu}
    \begin{tabular}{c c c c c}
        \toprule
        $\nu$ & \makecell{Data Queue \\ (Mb)} & \makecell{Energy Queue \\ (mJ)} & \makecell{Weighted Rate \\ (Mbps)} & \makecell{Power \\ (W)} \\
        \midrule
        1    & 10.752   & 234.061 & 37.35 & 0.0919 \\
        50   & 15.128   & 6.857   & 37.35 & 0.0775 \\
        100  & 27.708   & 6.680   & 37.35 & 0.0789 \\
        500  & 988.697  & 35.035  & 37.52 & 0.0800 \\
        1000 & 2838.127 & 76.044  & 35.31 & 0.0788 \\
        \bottomrule
    \end{tabular}
\end{table}

\section{CONCLUSION}\label{sec:CONCLUSION}

This paper proposes a complete low-complexity solution to the online optimization problem of joint offloading decision and resource allocation in quantum-secure NOMA systems. The research first explicitly incorporates the constant circuit power consumption of the post-quantum security module into the multi-user optimization framework, establishing a multi-stage MINLP model encompassing binary offloading decisions, local CPU frequency control, and NOMA transmission power allocation. By introducing a virtual energy queue, the long-term average energy consumption constraint is transformed into an online, trackable queue-stability problem, providing a theoretical foundation for I system optimization in quantum-secure scenarios. Subsequently, a NOMA power allocation algorithm based on greedy backward induction is proposed. By deriving a recursive analytical expression for the inter-user interference term, the global nonconvex power allocation problem is strictly decomposed into multiple independent single-variable convex subproblems with closed-form solutions. The computational complexity of this algorithm is only O(N), achieving significant acceleration over the SCA benchmark and meeting real-time requirements for millisecond-level online decision-making in fast-fading channels.

\bibliographystyle{IEEEtran}
\bibliography{references}

\end{document}